\documentclass[
  aps,
  prl,
  reprint,
  superscriptaddress,
  amsfonts,amssymb,amsmath
]{revtex4-2}
\usepackage{newtxtext}
\usepackage{bm}
\usepackage{latexsym}
\usepackage{amsthm}
\usepackage{mathtools}
\usepackage{xcolor}
\usepackage{physics}
\usepackage{graphicx}
\usepackage{hyperref}
\hypersetup{
  setpagesize  = false,
  colorlinks   = true,
  urlcolor     = blue,
  linkcolor    = blue,
  citecolor    = blue
}

\theoremstyle{plain}
\newtheorem{theorem}{Theorem}[section]
\newtheorem{lemma}[theorem]{Lemma}
\newtheorem{corollary}[theorem]{Corollary}
\newtheorem{proposition}[theorem]{Proposition}
\renewcommand{\thetheorem}{\arabic{theorem}}

\newcommand{\beginendmatt}{
  \twocolumngrid
  \setcounter{equation}{0}
  \renewcommand{\theequation}{A\arabic{equation}}%
  \renewcommand{\theHequation}{A.\arabic{equation}}%
  \setcounter{figure}{0}
  \renewcommand{\thefigure}{A\arabic{figure}}%
  \renewcommand{\theHfigure}{A.\arabic{figure}}%
  \setcounter{table}{0}
  \renewcommand{\thetable}{A\arabic{table}}%
  \renewcommand{\theHtable}{A.\arabic{table}}%
  \setcounter{section}{0}
  \renewcommand{\thesection}{A\Roman{section}}%
  \renewcommand{\theHsection}{A.\Roman{section}}%
}

\newcommand{\beginsupplement}{
  \setcounter{secnumdepth}{2}
  \setcounter{equation}{0}
  \renewcommand{\theequation}{S\arabic{equation}}%
  \renewcommand{\theHequation}{S.\arabic{equation}}%
  \setcounter{figure}{0}
  \renewcommand{\thefigure}{S\arabic{figure}}%
  \renewcommand{\theHfigure}{S.\arabic{figure}}%
  \setcounter{table}{0}
  \renewcommand{\thetable}{S\arabic{table}}%
  \renewcommand{\theHtable}{S.\arabic{table}}%
  \setcounter{section}{0}
  \renewcommand{\thesection}{S\Roman{section}}%
  \renewcommand{\theHsection}{S.\Roman{section}}%
  \renewcommand{\thetheorem}{\thesection.\arabic{theorem}}%
}

\begin{document}

%%% title %%%
\title{Universal Bound and Phase Transition in Many-Body Fermionic Non-Gaussianity}

%%% authors %%%
\author{Masahiro Hoshino}
% \email{hoshino-masahiro921@g.ecc.u-tokyo.ac.jp}
\affiliation{Department of Physics, The University of Tokyo, 7-3-1 Hongo, Bunkyo-ku, Tokyo 113-0033, Japan}

\author{Ryota Matsuda}
% \email{r-matsuda918@g.ecc.u-tokyo.ac.jp}
\affiliation{Department of Physics, The University of Tokyo, 7-3-1 Hongo, Bunkyo-ku, Tokyo 113-0033, Japan}

\author{Yuto Ashida}
% \email{ashida@phys.s.u-tokyo.ac.jp}
\affiliation{Department of Physics, The University of Tokyo, 7-3-1 Hongo, Bunkyo-ku, Tokyo 113-0033, Japan}
\affiliation{Institute for Physics of Intelligence, The University of Tokyo, 7-3-1 Hongo, Bunkyo-ku, Tokyo 113-0033, Japan}

%%% date %%% 
% \date{\today}

%%% abstract %%%
\begin{abstract}
  Fermionic non-Gaussianity is a resource for universal quantum computation that can be generated by interactions in quantum many-body systems.
  Using the magic R\'enyi entropy (MRE) as a measure of non-Gaussianity, we derive its universal upper bound and Haar mean, proving that typical Haar-random states attain the maximal MRE density of $\ln(4/3)$ per Majorana in the large-system limit.
  We also show that the MRE can differ extensively between states with identical covariance matrices.
  To investigate how non-Gaussianity grows toward this maximal density, we evolve a Gaussian state in imaginary and real time under the Sachdev-Ye-Kitaev Hamiltonian.
  By varying the imaginary- and real-time durations, we identify a first-order transition marked by a kink in the MRE density and spontaneous breaking of the permutation symmetry among the four copies used to evaluate the MRE.
  This transition represents a qualitative change in the non-Gaussianity of the state that is not reflected in the thermal free energy.
\end{abstract}

\maketitle

%%%========== main text ==========%%%

%%%========== Introduction ==========%%%
Fermionic non-Gaussianity, also known as fermionic magic, is a resource for quantum computation that can be generated by interactions in many-body systems.
Fermionic Gaussian states and Gaussian operations can be simulated efficiently on a classical computer by tracking the covariance matrix~\cite{terhal2002classical,bravyi2005lagrangian,dias2024classical,reardonsmith2024improved}.
Universal quantum computation can be achieved by supplementing Gaussian operations with non-Gaussian input states, for which any pure non-Gaussian state with definite fermion parity is sufficient~\cite{hebenstreit2019all}.
This parallels the role of nonstabilizerness, or magic, in qubit systems, where suitable magic states supplement classically simulable Clifford operations to enable universal quantum computation~\cite{aaronson2004improved,bravyi2005universal}.
Because non-Gaussianity cannot be generated by quadratic Hamiltonians, its growth in many-body dynamics directly reflects the effects of interactions~\cite{lumia2024measurement,paviglianiti2026emergence,falcao2026fermionic,aditya2026growth}.
Non-Gaussianity is also tied to electronic correlations because higher-order correlation functions in non-Gaussian states cannot all be reduced to two-point functions by Wick's theorem~\cite{gottlieb2007properties,pachos2022quantifying}.
Strongly correlated electron systems are therefore a natural setting in which to study how much non-Gaussianity can be generated by interactions.

Advances in measures of quantum computational resources and in methods for evaluating them have enabled quantitative studies of these resources in many-body systems~\cite{haug2023quantifying,lami2023nonstabilizerness,tarabunga2024nonstabilizerness,haug2024efficient,liu2025nonequilibrium,ding2025evaluating,paviglianiti2025estimating}.
For nonstabilizerness, studies of the stabilizer R\'enyi entropy (SRE)~\cite{leone2022stabilizer} have revealed its typical values and concentration in Haar-random states~\cite{liu2022many,turkeshi2025pauli,szombathy2025asymptotically}, universal behavior in quantum critical systems~\cite{hoshino2026stabilizer,hoshino2026stabilizera}, transitions in gauge theories~\cite{tarabunga2023many}, and its growth and saturation under time evolution~\cite{haug2025probing,turkeshi2025magic,tirrito2025anticoncentration}.
For fermionic non-Gaussianity, measures based on the covariance matrix have been used to study state conversion under Gaussian operations, classical simulation, and the dynamics of non-Gaussianity in many-body systems~\cite{gigena2015entanglement,gigena2020one,sierant2026fermionic,tarabunga2026computable,leone2026fermionic,turkeshi2026williamson,garciagarcia2026tuning}.
Another approach quantifies non-Gaussianity through higher-order correlations captured by a convolution operation, which mixes two copies of a state using a Gaussian unitary and traces out one of them~\cite{lyu2024fermionic,coffman2025measuring}.
Within this framework, the magic R\'enyi entropy (MRE) has been proposed and shown to exhibit universal behavior at quantum criticality described by conformal field theory~\cite{matsuda2026quantum}.

Despite these advances, it remains unclear how large the MRE of typical states is and whether its growth under interacting dynamics exhibits collective phenomena such as phase transitions.
Haar-random states serve as a standard reference for the typical properties of quantum states, as in studies of the mean and concentration of entanglement~\cite{page1993average,hayden2006aspects}.
Typical non-Gaussianity in subsystems of Haar-random states has also been characterized using measures based on relative entropy~\cite{ares2026non}.
For the MRE, however, the relation between its typical value and a universal upper bound has not been established.
In particular, it is not known whether typical states attain the maximal MRE per Majorana in the large-system limit.
Establishing this relation would provide a reference against which the non-Gaussianity generated by interactions can be compared.
The Sachdev-Ye-Kitaev (SYK) model provides a useful setting for studying how interactions generate non-Gaussianity, as it is strongly interacting and exhibits quantum chaos while remaining tractable in the limit of a large number of fermions~\cite{sachdev1993gapless,kitaev2015simple,maldacena2016remarks}.
Previous studies have investigated the SRE in this model~\cite{bera2025non,jasser2025stabilizer,russomanno2025nonstabilizerness,bettaque2026magic}, revealing transitions under time evolution~\cite{sun2026connecting} and, in coupled SYK models, transitions that are not detected by thermodynamic quantities~\cite{zhang2026stabilizer}.
These studies raise the questions of how the non-Gaussianity generated in such dynamics compares with its typical value and whether its growth involves a phase transition.

In this Letter, we prove that the MRE of Haar-random states with fixed parity concentrates around its mean and derive a universal upper bound, showing that typical states attain the maximal MRE density $\ln(4/3)$ per Majorana in the large-system limit.
We also construct a family of states whose MRE exceeds the Haar mean and show that the MRE can differ extensively between states with identical covariance matrices.
In the SYK model, our large-$N$ saddle-point analysis predicts a first-order transition in the MRE density as the durations of imaginary- and real-time evolution from a Gaussian state are varied.
This transition corresponds to spontaneous breaking of the permutation symmetry among the four copies used to evaluate the MRE, from $S_4$ to a cyclic subgroup $C_3$, and cannot be captured by the thermal free energy.

%%%========== MRE of Haar-random states ==========%%%
\textit{Non-Gaussianity of Haar-random states.}---We first introduce a convolution-based non-Gaussianity measure, which we call the MRE, and define the setting used throughout this Letter.
We consider $N=2n\geq2$ Majorana operators satisfying $\gamma_i^\dag=\gamma_i$ and $\{\gamma_i,\gamma_j\}=2\delta_{ij}$ for $i,j\in[N]=\{1,2,\ldots,N\}$.
The associated complex fermions are $f_m=(\gamma_{2m-1}-i\gamma_{2m})/2$, with $m\in[n]$, and the Fock space has dimension $D=2^n$.
We restrict attention to normalized pure states $\rho=\ketbra{\psi}{\psi}$ in the even-parity sector $\mathcal H_+$ of dimension $d=D/2$, i.e., the $+1$ eigenspace of the parity operator $\Pi=i^n\gamma_1\gamma_2\cdots\gamma_N$.
When several copies are involved, we label them by superscripts and take Majorana operators on different copies to anticommute, $\{\gamma_i^{(a)},\gamma_j^{(b)}\}=2\delta_{ij}\delta_{ab}$.
Following Refs.~\cite{lyu2024fermionic,coffman2025measuring}, we define the MRE through the purity of a convolution of two identical copies of the state.
To form the convolution, we mix the copies with the balanced beam splitter $U=\exp[(\pi/8)\sum_i\gamma_i^{(1)}\gamma_i^{(2)}]$ and trace out the second copy, obtaining $\mathcal C(\rho)=\Tr_2[U\rho^{\otimes2}U^\dag]$.
The two output copies remain unentangled, i.e., $\mathcal C(\rho)$ is pure, if and only if $\ket\psi$ is Gaussian.
The MRE is then given by
\begin{equation}\label{eq:convolution-purity}
  M_2(\psi) =-\ln\nu_2(\psi),\quad \nu_2(\psi)=\Tr[\mathcal C(\rho)^2],
\end{equation}
which is nonnegative and invariant under any Gaussian unitary.
It is also additive, $M_2(\psi\otimes\phi)=M_2(\psi)+M_2(\phi)$, and monotone, i.e., it does not increase under deterministic Gaussian protocols that map pure states to pure states~\cite{matsuda2026quantum}.

Let $U_{ab}$ denote the balanced beam splitter on copies $a,b$.
Applying $U_{12}U_{34}$ to four copies and using the swap trick on the retained copies 1 and 3 gives
\begin{equation}
  \nu_2(\psi)=\Tr(\rho^{\otimes4}\mathcal R)
  =\bra\psi^{\otimes4}\mathcal R\ket\psi^{\otimes4}.
  \label{eq:purity-representations}
\end{equation}
Here, $\mathcal R=\Pi_3\Pi_4(U_{12}U_{34})^\dag S_{13}(U_{12}U_{34})\Pi_3\Pi_4$, where $\Pi_a=i^n\gamma_1^{(a)}\gamma_2^{(a)}\cdots\gamma_N^{(a)}$ is the parity of copy $a$ and the fermionic SWAP $S_{ab}$ exchanges the Majorana operators of copies $a$ and $b$ while leaving the Fock vacuum invariant.
Because both retained copies have even parity for identical inputs in $\mathcal H_+$, the ordinary SWAP can be replaced by $S_{13}$, which makes $\mathcal R$ Gaussian.
The input parities also allow conjugation by $\Pi_3\Pi_4$ without changing the expectation value; this makes the action of $\mathcal R$ permutation symmetric.
Equation~\eqref{eq:purity-representations} implies that the first and second moments of the purity are determined by the fourth and eighth Haar moments of $\rho$, respectively.

We compute the mean purity using the moment formula for Haar-random states and determinant formulas for traces of Gaussian unitaries.
Writing $\mathbb E_{\mathrm H}$ for the mean over the complex Haar measure on normalized states in $\mathcal H_+$ and $\bar\nu=\mathbb E_{\mathrm H}\nu_2$, we use $\mathbb E_{\mathrm H}\rho^{\otimes4}=(d)_4^{-1}\sum_{\pi\in S_4}\mathsf P_\pi$, where $(d)_4=d(d+1)(d+2)(d+3)$, $S_4$ is the permutation group of four objects, and $\mathsf P_\pi$ permutes the factors of $\mathcal H_+^{\otimes4}$.
Expanding the even-parity projectors reduces $\operatorname{Tr}(\mathsf P_\pi\mathcal R)$ to determinant sums [see Supplemental Material (SM)], yielding
\begin{equation}
  \bar\nu=\frac{8\cdot3^N+12D^3+16D^2+48D+56}
  {D(D+2)(D+4)(D+6)}.
  \label{eq:haar-mean-purity}
\end{equation}
For large $N$, the numerator is dominated by $8\cdot3^N$, so that $\bar\nu=8(3/4)^N[1+o(1)]$.
Because the logarithm does not commute with the average, obtaining the mean MRE $\overline{M_2}=\mathbb E_{\mathrm H}[-\ln\nu_2]$ also requires bounding the Jensen gap $\overline{M_2}+\ln\bar\nu\geq0$.

Using the eighth Haar moment, we show that this gap vanishes as the system size increases and that the MRE concentrates around its mean.
Specifically, the relative variance is $V_N=\operatorname{Var}_{\mathrm H}(\nu_2)/\bar\nu^2=144(5/9)^N[1+o(1)]$ (see SM for the derivation).
Chebyshev's inequality then gives $\Pr(|\nu_2/\bar\nu-1|>\epsilon)\leq V_N/\epsilon^2$ for any fixed $\epsilon>0$.
Hence $\nu_2/\bar\nu\to1$ in probability, and by continuity of the logarithm, $M_2+\ln\bar\nu\to0$ in probability.
Together with the bound $\nu_2\geq1/d$, the relative variance also bounds the Jensen gap as $0\leq\overline{M_2}+\ln\bar\nu\leq[1+4\ln(d\bar\nu)]V_N=\mathcal{O}(NV_N)$.
Since $NV_N\to0$, the mean MRE has the asymptotic form
\begin{equation}
  \overline{M_2}=N\ln\frac43-\ln8+o(1).
  \label{eq:haar-mean-mre}
\end{equation}
Thus, the MRE density per Majorana approaches $\ln(4/3)$.

%%%========== Upper bound on the MRE ==========%%%
\textit{Universal upper bound on the MRE.}---To assess how close the MRE of Haar-random states is to the largest possible value, we prove the following universal bound.
\begin{theorem}
  For any even $N\geq2$, every normalized pure state $\ket\psi\in\mathcal H_+$ satisfies
  \begin{equation}
    M_2(\psi)\leq N\ln\frac43-\ln\frac{16}{9}.
    \label{eq:mre-upper-bound}
  \end{equation}
\end{theorem}
This bound exceeds the Haar mean in Eq.~\eqref{eq:haar-mean-mre} only by $\ln(9/2)+o(1)$.
Hence the maximal MRE density is $\ln(4/3)$ in the large-$N$ limit, and typical Haar-random states attain it.
Whether the bound is tight at finite $N$ remains open.

We prove the bound by decomposing the beam splitter $U$ into $N$ commuting rotations, one for each Majorana index, and combining a one-step inequality for the purity with an exact identity for the last rotation.
For $\Lambda\subseteq[N]$, we define the partial convolution $U_\Lambda=\prod_{i\in\Lambda}\exp[(\pi/8)\gamma_i^{(1)}\gamma_i^{(2)}]$ and denote by $\nu_\Lambda$ the purity after applying $U_\Lambda$ to $\ket\psi^{\otimes2}$ and tracing out the second copy, so that $\nu_\varnothing=1$ and $\nu_{[N]}=\nu_2(\psi)$.
To analyze the next rotation, labeled by $j\notin\Lambda$, we choose an orthonormal basis $\{\ket{e_\alpha}\}_{\alpha=1}^d$ of $\mathcal H_+$ and set $\ket{o_\alpha}=\gamma_j\ket{e_\alpha}$, which gives an orthonormal basis of the odd-parity sector.
Since the total parity remains even, the two copies always have the same parity, and the partially convolved state can be written as $U_\Lambda\ket\psi^{\otimes2}=\sum_{\alpha,\beta=1}^d E_{\alpha\beta}\ket{e_\alpha}_1\ket{e_\beta}_2+\sum_{\alpha,\beta=1}^d O_{\alpha\beta}\ket{o_\alpha}_1\ket{o_\beta}_2$, where $E$ and $O$ are the corresponding coefficient matrices.
The reduced state is block diagonal with blocks $EE^\dag$ and $OO^\dag$, so $\nu_\Lambda=\Tr[(EE^\dag)^2]+\Tr[(OO^\dag)^2]$.
The rotation labeled by $j$ then acts as $E'=cE-sO$ and $O'=sE+cO$, where $c=\cos(\pi/8)$ and $s=\sin(\pi/8)$.
Using the fact that the two inputs are identical pure states, we show in the SM that
\begin{align}
  \nu_{\Lambda\cup\{j\}}-\frac34\nu_\Lambda
  =\frac52\|O^\dag E\|_{\mathrm F}^2            +\frac14\|EO^\dag+OE^\dag\|_{\mathrm F}^2\geq0,
  \label{eq:recursive-purity-bound}
\end{align}
where $\|A\|_{\mathrm F}^2=\Tr(A^\dag A)$ is the squared Frobenius norm.
A single rotation therefore cannot reduce the purity below $3/4$ of its previous value.
For the last label $j$, both copies have even parity after the final rotation, as noted above, so $O'=0$ and hence $O=-(s/c)E$, which gives $\nu_{[N]\setminus\{j\}}=(c^4+s^4)\nu_{[N]}=(3/4)\nu_{[N]}$.
Starting from $\nu_\varnothing=1$ and applying Eq.~\eqref{eq:recursive-purity-bound} to the first $N-1$ rotations, we obtain
\begin{equation}
  \nu_2(\psi)=\frac43\nu_{[N]\setminus\{j\}}
  \geq\frac43\left(\frac34\right)^{N-1}
  =\left(\frac34\right)^{N-2}.
  \label{eq:terminal-purity-bound}
\end{equation}
Equation~\eqref{eq:mre-upper-bound} follows by taking the negative logarithm.
Analogous Haar means and upper bounds for higher-order MREs, defined through convolutions of three or more copies, are given in the SM.

%%%========== States with large MRE and comparison with measures based on covariance matrices ==========%%%
\textit{Families of non-Gaussian states with identical covariance matrices.}---Typical Haar-random states attain the maximal MRE density in the large-$N$ limit, but they do not maximize the MRE itself: we can deterministically construct a family of states whose MRE exceeds the Haar mean for infinitely many system sizes.
For $N\in8\mathbb N$, let $K$ denote complex conjugation in the occupation basis of the fermions $f_m$.
The antiunitary operator $T=\gamma_1\gamma_3\cdots\gamma_{2n-1}K$ satisfies $T^2=1$ and preserves $\mathcal H_+$.
We call a state a real spinor if $T\ket\psi=\ket\psi$.
\begin{theorem}
  \label{thm:real-spinors-above-haar}
  For infinitely many $N\in8\mathbb N$, there exists a normalized real spinor $\ket{\psi_N}\in\mathcal H_+$ such that $M_2(\psi_N)>\overline{M_2}$.
  Such states can be constructed by a finite deterministic procedure.
\end{theorem}

Let $\bar\nu_{\mathbb R}=\mathbb E_{\mathbb R}\nu_2$ denote the mean purity over uniformly distributed normalized real spinors.
With $\alpha=\arctan(1/2)$, we find
\begin{equation}
  \frac{\bar\nu-\bar\nu_{\mathbb R}}{\bar\nu}
  =-12\left(\frac{\sqrt5}{3}\right)^N\cos(N\alpha)
  +\mathcal{O}\!\left(2^{-N/2}\right).
  \label{eq:real-haar-difference}
\end{equation}
Since $\alpha/\pi$ is irrational, infinitely many $N\in8\mathbb N$ satisfy $\cos(N\alpha)\leq-1/2$.
Along this sequence, the positive leading term dominates both the remainder and the Haar Jensen gap $\overline{M_2}+\ln\bar\nu=\mathcal{O}(N(5/9)^N)$.
Hence $\mathbb E_{\mathbb R}M_2\ge-\ln\bar\nu_{\mathbb R}>\overline{M_2}$ for sufficiently large $N$, so some real spinor has MRE above the Haar mean.
This averaging argument can be turned into a deterministic construction of such states (see SM).

The states constructed above also allow us to compare the MRE with measures of fermionic non-Gaussianity based on the covariance matrix.
We consider the fermionic antiflatness (FAF)~\cite{sierant2026fermionic}, defined by $\mathcal F_k(\psi)=n-\tfrac12\Tr[(\Gamma^{\top}\Gamma)^k]$ for integers $k\geq1$.
Here, $\Gamma_{ij}=(i/2)\bra\psi[\gamma_i,\gamma_j]\ket\psi$ is the covariance matrix of $\ket\psi$.
For $k=1$, the FAF is proportional to the occupation-number entropy of Tsallis order two introduced in Ref.~\cite{tarabunga2026computable} and is a strong monotone under Gaussian protocols.
The FAF attains its maximal value $\mathcal F_k=n$ if and only if $\Gamma=0$.
For example, the cat state $\ket{\mathrm{cat}_n}=(\ket{0^n}+\ket{1^n})/\sqrt2$, a superposition of the empty and fully occupied Fock states, has $\Gamma=0$ for even $n\geq4$.
In contrast, the MRE of these cat states approaches the constant $\ln8$ as $n\to\infty$, since $\nu_2(\mathrm{cat}_n)=1/8+7/2^{n+1}$.
The real spinors $\ket{\psi_N}$ in Theorem~\ref{thm:real-spinors-above-haar} also have $\Gamma=0$, since $T$ invariance and $T(i\gamma_i\gamma_j)T^{-1}=-i\gamma_i\gamma_j$ for $i\ne j$ force the real entries $\Gamma_{ij}$ to vanish.
Since their MRE exceeds the Haar mean and obeys the universal upper bound, we have $M_2(\psi_N)=N\ln(4/3)+\mathcal{O}(1)$.
For this sequence of system sizes, we therefore obtain
\begin{align}
   & \mathcal F_k(\psi_N)=\mathcal F_k(\mathrm{cat}_{N/2})=\frac N2,
  \label{eq:equal-faf}                                                 \\
   & M_2(\psi_N)-M_2(\mathrm{cat}_{N/2})  =N\ln\frac43+\mathcal{O}(1).
  \label{eq:equal-faf-mre-gap}
\end{align}
Thus, the MRE detects extensive differences in non-Gaussianity between states that are indistinguishable by covariance-based measures.

Conversely, we note that states with the same MRE can have different values of the FAF.
For example, on six complex fermion modes, consider $\ket\chi=\ket{\mathrm{cat}_4}\otimes\ket{00}$ and $\ket\eta=\sqrt{(1+z)/2}\ket{0^6}+\sqrt{(1-z)/2}\ket{1^6}$.
One can tune $z=\sqrt{2\sqrt{65}/5-3}$ to satisfy $M_2(\chi)=M_2(\eta)=\ln(32/11)$ while $\mathcal F_k(\chi)=4<6(1-z^{2k})=\mathcal F_k(\eta)$ for every integer $k\geq1$.
Together with the previous example, this shows that neither the MRE nor the FAF determines the other.
Because both the MRE and the FAF at $k=1$ are monotones under Gaussian protocols, they impose complementary constraints on deterministic conversion between pure states (see Ref.~\cite{tarabunga2026bell} for related discussions).

%%%========== MRE dynamics and transition in the SYK model ==========%%%
\textit{Dynamics and transition of non-Gaussianity in the SYK model.}---Using the SYK model as a concrete example, we now demonstrate how chaotic dynamics can generate almost maximal non-Gaussianity.
The Hamiltonian of the SYK model is
\begin{equation}\label{eq:syk-hamiltonian}
  H = -\sum_{1\leq i<j<k<l\leq N}J_{ijkl}\gamma_i\gamma_j\gamma_k\gamma_l,
\end{equation}
where the couplings $J_{ijkl}$ are independent Gaussian random variables with zero mean and variance $6J^2/N^3$, and $J>0$ sets the energy scale.
We denote the average over the couplings by $\mathbb E_J$.
We evaluate the MRE of the Kourkoulou-Maldacena (KM) state~\cite{kourkoulou2017pure,zhang2020entanglement}, defined by
\begin{equation}
  \ket*{\psi(\beta,t)}=\frac{e^{-wH}\ket B}{\sqrt{Z_B(\beta)}}.
\end{equation}
Here, $\ket B$ is the Fock vacuum, a Gaussian state, and $w=\beta/2+it$ describes imaginary-time evolution by $\beta/2$ followed by real-time evolution by $t$.
The normalization factor $Z_B(\beta)=\bra B e^{-\beta H}\ket B$ is the partition function on an imaginary-time interval of length $\beta$, with boundary state $\ket B$ at both ends.
We study the MRE density $M_2/N$ of the KM state and its possible nonanalyticities in the large-$N$ limit.

For the KM state, the four-copy representation in Eq.~\eqref{eq:purity-representations}, which we used above for the Haar mean, gives the purity as $\nu_2(\beta,t)=Z_R/Z_B(\beta)^4$ with
\begin{equation}\label{eq:syk-amplitude}
  Z_R            =\bra B^{\otimes4}(e^{-\bar wH})^{\otimes4}\mathcal R\,(e^{-wH})^{\otimes4}\ket B^{\otimes4}.
\end{equation}
Here, $\bar w$ is the complex conjugate of $w$, and all four copies share the same couplings $J_{ijkl}$.
The unitary $\mathcal R$ acts on the Majorana operators as $\mathcal R\gamma_i^{(a)}\mathcal R^\dagger=\sum_{b=1}^4R_{ab}\gamma_i^{(b)}$, where the reflection $R=I_4-\bm1\bm1^{\top}/2$ with $\bm1=(1,1,1,1)^{\top}$ commutes with permutations of the copies.
To write $Z_R$ as a path integral that incorporates the reflection, we introduce a complex-time contour $z(u)$, with $u\in[0,2]$ and $z(0)=0$.
We choose $\dot z=\mathrm dz/\mathrm du=w$ for $0<u<1$ and $\dot z=\bar w$ for $1<u<2$, corresponding to evolution before and after the insertion of $\mathcal R$ at $u=1$.
We use the same symbol $\gamma_i^{(a)}(u)$ for the Grassmann field on the contour.
To derive the saddle-point equations, we introduce the bilocal field
\begin{equation}
  \widehat G_{ab}(u,v)=\frac1N\sum_{i=1}^N\gamma_i^{(a)}(u)\gamma_i^{(b)}(v).
  \label{eq:syk-collective-field}
\end{equation}
At a saddle, $\widehat G$ represents the contour-ordered two-point function of the Majorana fields.
The effect of the reflection can be absorbed into the interaction term by introducing the rotated field $G(u,v)=O(u)^{\top}\widehat G(u,v)O(v)$, where
\begin{equation}
  O(u)=\begin{cases}
    I_4, & 0\leq u<1, \\
    R,   & 1<u\leq2.
  \end{cases}
  \label{eq:syk-copy-rotation}
\end{equation}
We introduce a Lagrange multiplier $\Sigma_{ab}(u,v)=-\Sigma_{ba}(v,u)$, average over disorder, and integrate out the fermions to obtain $\mathbb E_JZ_R=\int \mathcal DG\,\mathcal D\Sigma\,e^{-NI[G,\Sigma]}$, with the effective action per Majorana
\begin{align}
  I[G,\Sigma]={} & -\frac12\ln\det(1-2F_0\Sigma)
  \notag                                                                                               \\
                 & +\frac12\int\dd u\,\dd v\sum_{a,b}\Sigma_{ab}(u,v)G_{ab}(u,v)
  \notag                                                                                               \\
                 & -\frac{J^2}{8}\int\dd u\,\dd v\,\dot z(u)\dot z(v)\sum_{a,b}\widehat G_{ab}(u,v)^4.
  \label{eq:syk-action}
\end{align}
Here, $F_0(u,v)=\Theta(u-v)I_4$ is the free-particle propagator fixed by the vacuum boundary conditions at both contour endpoints, where $\Theta$ is the step function with $\Theta(0)=1/2$.
Products such as $F_0\Sigma$ denote composition of kernels over copy indices and contour coordinates, and $1$ is the identity kernel $\delta_{ab}\delta(u-v)$.
Varying $I$ with respect to $G$ and $\Sigma$ gives the Schwinger-Dyson (SD) equations (see SM for the derivation):
\begin{align}
  F           & =(1-2F_0\Sigma)^{-1}F_0,
  \label{eq:syk-dyson}                                                                     \\
  G           & =F-F^{\top},
  \label{eq:syk-majorana-propagator}                                                       \\
  \Sigma(u,v) & =J^2\dot z(u)\dot z(v)O(u)^{\top}\bigl[\widehat G(u,v)^{\circ3}\bigr]O(v).
  \label{eq:syk-self-energy}
\end{align}
Here, $F$ is the dressed particle propagator, and the Majorana propagator $G$ combines the particle and hole contributions $F$ and $-F^{\top}$.
The transpose exchanges copy indices and contour coordinates without complex conjugation, $F^{\top}_{ab}(u,v)=F_{ba}(v,u)$, and $\circ3$ denotes the elementwise cube.

\begin{figure*}[t]
  \centering
  \includegraphics[width=0.92\textwidth]{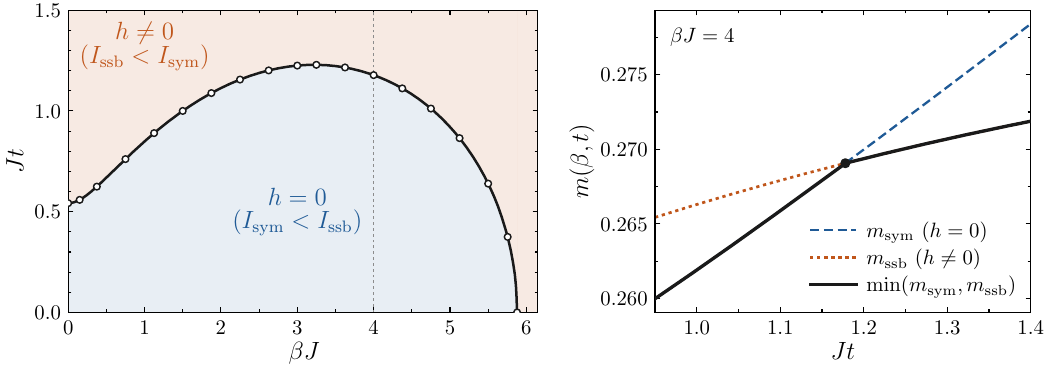}
  \caption{Phase diagram (left) and MRE density (right) of the KM state in the large-$N$ SYK model, obtained by numerically solving the SD equations with $L=512$ integration points, without extrapolation to the continuum limit.
    Left: regions where the symmetric saddle ($h=0$, blue shading) or the saddle with broken symmetry ($h\neq0$, orange shading) has the smaller action.
    Open circles mark crossing points of the two actions, and the solid curve gives the phase boundary.
    The vertical line with short dashes marks the cut at $\beta J=4$ shown in the right panel.
    Right: $m_{\mathrm{sym}}$ (blue dashed) and $m_{\mathrm{ssb}}$ (orange dotted), together with their minimum $m(\beta,t)$ (solid black), as functions of $Jt$ at $\beta J=4$, with $m(\beta,t)$ defined in Eq.~\eqref{eq:syk-minimum-density}.
    The black circle marks the transition between the two saddles, where $m(\beta,t)$ has a kink.}
  \label{fig:syk-transition}
\end{figure*}

To solve the SD equations, we use the ansatz
\begin{equation}\label{eq:syk-saddle-ansatz}
  G = gI_4 + hA, \quad A=\mqty(         0  & 1  & 1  & 1  \\
  -1 & 0  & 1  & -1 \\
  -1 & -1 & 0  & 1  \\
  -1 & 1  & -1 & 0).
\end{equation}
This ansatz is preserved by both the self-energy equation [Eq.~\eqref{eq:syk-self-energy}] and Dyson's equation [Eq.~\eqref{eq:syk-dyson}], reducing the problem to two scalar functions with $g(v,u)=-g(u,v)$ and $h(v,u)=h(u,v)$ (see End Matter).
The saddle with $h=0$ respects the permutation symmetry $S_4$, whereas a saddle with $h\neq0$ preserves only the subgroup $C_3$ generated by cyclic permutations of copies 2, 3, and 4.
Thus, $h$ is an order parameter for the breaking of the copy-permutation symmetry, and the eight saddles with $h\neq0$ are related to one another by permutations of the copies.

The competition between these two saddles suggests a phase transition in the MRE density of the KM state in the large-$N$ limit.
Let $I_{\mathrm{sym}}$ and $I_{\mathrm{ssb}}$ denote the real parts of the actions at the saddles with $h=0$ and $h\neq0$, respectively.
We assume that these saddles determine the leading exponential behavior of $\mathbb E_JZ_R$ in the large-$N$ limit (see SM for numerical checks).
The normalization obeys $\mathbb E_JZ_B=2^{-N/2}\mathbb E_J\operatorname{Tr}e^{-\beta H}=e^{-NI_1+o(N)}$~\cite{kourkoulou2017pure}, where $I_1$ is determined by the thermal free energy.
The resulting annealed MRE density is
\begin{align}
  m(\beta,t) & \equiv-\lim_{N\to\infty}\frac1N\ln\frac{\mathbb E_JZ_R}{(\mathbb E_JZ_B)^4}
  \notag                                                                                   \\
             & =\min\{m_{\mathrm{sym}}(\beta,t),m_{\mathrm{ssb}}(\beta,t)\},
  \label{eq:syk-minimum-density}
\end{align}
where $m_{\mathrm{sym}}=I_{\mathrm{sym}}-4I_1$ and $m_{\mathrm{ssb}}=I_{\mathrm{ssb}}-4I_1$.
This density equals both the disorder-averaged value and the MRE density of a typical realization if $\ln\mathbb E_JZ_R-\mathbb E_J\ln Z_R=o(N)$; the corresponding condition for $Z_B$ is proved in the SM.
Numerically solving the SD equations gives the phase diagram in Fig.~\ref{fig:syk-transition}.

The phase boundary marks a first-order transition, at which the order parameter $h$ jumps and the MRE density $m$ has a kink because the two actions cross with different slopes.
Both adjacent phases have extensive non-Gaussianity.
Here, the upper bound~\eqref{eq:mre-upper-bound} on the MRE imposes a nontrivial constraint on the phase diagram, indicating that the symmetric saddle cannot remain dominant throughout.
For example, at $t=0$, following the symmetric saddle as $\beta$ increases eventually gives a density exceeding the maximal value $\ln(4/3)$ imposed by the universal upper bound in Eq.~\eqref{eq:mre-upper-bound}, which also applies to $m(\beta,t)$ since $\mathbb E_JZ_B^4\geq(\mathbb E_JZ_B)^4$.
We find that the saddle with broken symmetry acquires the smaller action before the density of the symmetric saddle reaches this value, giving an MRE density consistent with the upper bound.
Since $I_1$ is analytic in $\beta$ and independent of $t$, the kink originates from $Z_R$ alone.
The MRE thus detects a transition in the KM state that is not reflected in the thermal free energy of the SYK model.

Interestingly, real- and imaginary-time evolutions modify the state in fundamentally different ways.
At fixed $\beta$, real-time evolution changes only the relative phases in the energy basis, so the energy distribution of each realization remains unchanged.
The transition at $\beta J=4$ in Fig.~\ref{fig:syk-transition} therefore arises entirely from these relative phases.
Increasing $\beta$ instead suppresses high-energy components and monotonically lowers the mean energy.
The MRE, however, depends on how $\mathcal R$ connects the energy eigenstates of the four copies, so lowering the energy need not move the state monotonically toward either phase.
Indeed, at fixed $Jt=1$, the KM state shows reentrant behavior as $\beta$ increases.
It passes from the phase with broken symmetry to the symmetric phase and back, although its mean energy decreases monotonically.

%%%========== Conclusion ==========%%%
\textit{Conclusion.}---We established a universal upper bound on the MRE and showed that typical Haar-random states attain the maximal MRE density $\ln(4/3)$ in the thermodynamic limit.
We also constructed a family of states whose MRE exceeds the Haar mean and showed that the MRE can differ extensively between states with identical covariance matrices, demonstrating that it probes non-Gaussianity beyond two-point correlations.
In the SYK model, the upper bound led to a nontrivial constraint on the non-Gaussianity phase diagram, and we numerically identified a first-order transition in the MRE density of the KM state associated with spontaneous breaking of the permutation symmetry among the four copies.
This transition, which is not reflected in the thermal free energy, shows that the growth of non-Gaussianity in chaotic dynamics can exhibit collective phenomena that are not readily apparent in conventional physical observables.

Several questions remain open.
An important question is whether analogous transitions occur in systems with spatially local interactions and how the growth of the MRE depends on transport and scrambling.
The connection between the SYK model and gravity~\cite{maldacena2016remarks,kourkoulou2017pure} also raises the possibility of a geometric interpretation of the MRE and its transition.
Finally, determining the typical MRE at fixed particle number would clarify how particle-number conservation constrains non-Gaussianity and provide a benchmark for its dynamics in such systems.

%%%=========== main text end ===========%%%

%%% acknowledgments %%%
\begin{acknowledgments}
  \textit{Acknowledgments.}---M.H. and R.M. were supported by FoPM, WINGS Program, the University of Tokyo.
  The work of M.H. was also supported by JSPS KAKENHI Grant No.~26KJ0910.
  Y.A. acknowledges support from JST FOREST Program (Grant No. JPMJFR222U), JST CREST (Grant No. JPMJCR23I2), and JST Moonshot R\&D (Grant No. JPMJMS256J).

  The authors used ChatGPT (OpenAI; models 5.6sol and 6astra) to assist with computing and checking the mean and variance under the Haar measure, developing the proof of the universal upper bound on the MRE, and numerically solving the Schwinger-Dyson equations for the SYK model.
  The authors directed these tasks and reviewed the resulting derivations and numerical results, and they take full responsibility for the content of this work.
\end{acknowledgments}

%%% bibliography %%%
\bibliography{biblio}

%%%=========== End Matter ===========%%%
% \clearpage
\onecolumngrid
\begin{center}
  \textbf{\Large End Matter}
\end{center}
\beginendmatt

\textit{Reduction to two scalar functions.}---The ansatz $G=gI_4+hA$ reduces the SD equations to equations for two scalar functions because both the self-energy equation and Dyson's equation preserve its dependence on the copy indices.
We first substitute this ansatz into Eq.~\eqref{eq:syk-self-energy} and evaluate the self-energy separately when $u$ and $v$ lie on the same contour segment and when they lie on different segments (the segments $0\leq u<1$ and $1<u\leq2$ describe the evolution before and after the insertion of $\mathcal R$).
When $u$ and $v$ lie on the same segment, the correlator $\widehat G$ entering the interaction is either $gI_4+hA$ or $gI_4+hRAR$.
Both $A$ and $RAR$ have vanishing diagonal entries and off-diagonal entries $\pm1$, so cubing each entry and rotating back produces only terms proportional to $I_4$ and $A$.
For $u<1<v$, the same operation is $[(gI_4+hA)R]^{\circ3}R$, whereas for $v<1<u$ it is $R[R(gI_4+hA)]^{\circ3}$.
These expressions mix $g$ and $h$ but still reduce to linear combinations of $I_4$ and $A$.
The self-energy therefore has the form $\Sigma=\sigma_I I_4+\sigma_A A$, with $\sigma_I=J^2\dot z(u)\dot z(v)g^3$ and $\sigma_A=J^2\dot z(u)\dot z(v)h^3$ on the same segment and
\begin{align}
  \sigma_I & =\frac{J^2\dot z(u)\dot z(v)}4(g^3+9gh^2-6\epsilon h^3),\notag \\
  \sigma_A & =\frac{J^2\dot z(u)\dot z(v)}4(3g^2h-6\epsilon gh^2+7h^3)
  \label{eq:em-reduced-self-energy}
\end{align}
for different segments, where $\epsilon=1$ for $u<1<v$ and $\epsilon=-1$ for $v<1<u$.
All products in these expressions for the self-energy are pointwise.

We next show that Dyson's equation maps this self-energy to a particle propagator with the same dependence on the copy indices.
The matrix $A$ obeys $A^2=-3I_4$ and has eigenvalues $\pm i\sqrt3$, so $\Sigma$ acts as $\sigma_I\pm i\sqrt3\sigma_A$ in the corresponding eigenspaces.
Since the free propagator $F_0=\Theta I_4$ does not mix these eigenspaces, Dyson's equation can be solved independently in each one to obtain the functions $f_\pm(u,v)$:
\begin{equation}
  f_\pm=[1-2\Theta(\sigma_I\pm i\sqrt3\sigma_A)]^{-1}\Theta.
  \label{eq:em-reduced-dyson}
\end{equation}
Products in this equation denote composition of kernels, with integration over intermediate contour coordinates.
The propagator $F=f_I I_4+f_A A$ reproduces the two solutions when $f_I\pm i\sqrt3\,f_A=f_\pm$ in the corresponding eigenspaces of $A$.
The sum and difference of these equations determine its coefficients:
\begin{equation}
  f_I=\frac{f_++f_-}{2},\qquad f_A=\frac{f_+-f_-}{2i\sqrt3}.
  \label{eq:em-particle-components}
\end{equation}
The relation $G=F-F^{\top}$ then recovers $G=gI_4+hA$, with $g=f_I-f_I^{\top}$ and $h=f_A+f_A^{\top}$ because $A^{\top}=-A$.
For the functions $f_I$ and $f_A$, the transpose exchanges $u$ and $v$ without complex conjugation.
Both the self-energy equation and Dyson's equation therefore preserve the form $G=gI_4+hA$, allowing us to solve for the two functions $g$ and $h$.

\textit{Numerical solution.}---We solve the SD equations for the two scalar functions $g$ and $h$ by replacing the contour integrals with finite sums and iterating between the self-energy equation and Dyson's equation.
We use the midpoint rule, dividing $[0,2]$ into $L$ equal subintervals with $L$ even and using their midpoints $u_k=(2k-1)/L$, $k=1,2,\ldots,L$, as integration points of weight $2/L$.
We represent each scalar function of $u$ and $v$ by an $L\times L$ matrix of its values at $(u_k,u_l)$, so Eq.~\eqref{eq:em-reduced-dyson} becomes
\begin{equation}
  \left[I_L-\frac{8}{L^2}\Theta(\sigma_I\pm i\sqrt3\sigma_A)\right]f_\pm=\Theta,
  \label{eq:em-discrete-dyson}
\end{equation}
where $\Theta_{kl}=\Theta(u_k-u_l)$ and $\Theta_{kk}=1/2$.
The factor $8/L^2=2(2/L)^2$ accounts for the two contour integrations in $2\Theta\Sigma F$.

\begin{figure}[tb]
  \centering
  \includegraphics[width=\columnwidth]{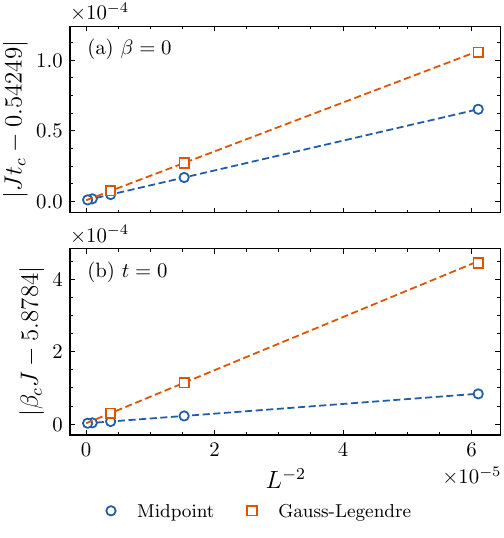}
  \caption{Convergence of the location of the phase boundary with the number $L$ of integration points on the contour.
    (a) Critical time $t_c$ at $\beta=0$ and (b) critical inverse temperature $\beta_c$ at $t=0$, shown as $|Jt_c-0.54249|$ and $|\beta_cJ-5.8784|$ versus $L^{-2}$, respectively, where $0.54249$ and $5.8784$ are the extrapolated continuum values.
    Circles indicate the midpoint rule with $L=128,256,512,1024,2048$, and squares indicate Gauss-Legendre quadrature with $L=128,256,512$.
    Dashed lines show linear extrapolations in $L^{-2}$ to the continuum limit $L^{-2}\to0$ from the two largest values of $L$ for each method.}
  \label{fig:em-convergence}
\end{figure}

We obtain a self-consistent solution by iterating between the self-energy equation and Dyson's equation.
From trial values of $f_I$ and $f_A$, we construct $g$ and $h$, calculate $\sigma_I$ and $\sigma_A$, and solve Eq.~\eqref{eq:em-discrete-dyson} for $f_\pm$, from which Eq.~\eqref{eq:em-particle-components} gives new values of $f_I$ and $f_A$.
We define the residual as the difference between the input values of $f_I$ and $f_A$ and those returned by one iteration.
We stabilize the iteration by averaging the input and returned values and accelerate convergence with Anderson mixing~\cite{anderson1965iterative,walker2011anderson}, a method that combines previous updates to reduce the estimated residual.
For every solution used to evaluate an action, we require the largest absolute entry of the residual, evaluated before mixing, to be below $10^{-11}$ for both $f_I$ and $f_A$.

We obtain the two competing saddles by first finding a solution of each type and then following both as $\beta$ and $t$ vary.
For the symmetric saddle, we start the iteration from the free propagator $f_I=\Theta$, $f_A=0$, which keeps $h=0$ throughout the calculation.
To find an initial solution with $h\neq0$, we temporarily modify the equations by replacing $F_0$ with $[\tfrac12\operatorname{sgn}(u-v)+\lambda/2]I_4$, where $\lambda=1$ recovers the original KM boundary condition.
We begin with $\lambda=2$ at $(\beta J,Jt)=(8,0)$ on a coarse grid and solve the SD equations for the full $4\times4$ matrix $G$ without imposing the ansatz, starting from the symmetric solution perturbed by weak correlations between pairs of copies.
Once a solution with $h\neq0$ is found, we gradually decrease $\lambda$ to $1$, still without imposing the ansatz, using each converged solution as the initial condition for the next value of $\lambda$.
At $\lambda=1$, we obtain a solution with $h\neq0$ that satisfies the original KM boundary condition, and we then increase $L$ and solve the reduced equations for $g$ and $h$ before evaluating its action.
We then vary $\beta$ and $t$ in small steps and use each converged solution to initialize the next calculation, following the $h=0$ and $h\neq0$ saddles separately through the crossing of their actions.

\textit{Determining the phase boundary.}---We evaluate and compare the actions of the two saddles to obtain the MRE density and locate the phase boundary.
At a saddle, we substitute the expression for $\Sigma$ from Eq.~\eqref{eq:syk-self-energy} into Eq.~\eqref{eq:syk-action} to combine its second and third terms.
On the discretized contour, the resulting action is
\begin{align}
  I={} & -\sum_\pm\ln\det\left[I_L-\frac{8}{L^2}\Theta(\sigma_I\pm i\sqrt3\sigma_A)\right]\notag \\
       & +\frac{3J^2}{2L^2}\sum_{k,l}\dot z(u_k)\dot z(u_l)\sum_{a,b}\widehat G_{ab}(u_k,u_l)^4.
  \label{eq:em-stationary-action}
\end{align}
Each eigenvalue of $A$ has multiplicity two, which cancels the factor $1/2$ in the logarithmic term of Eq.~\eqref{eq:syk-action}.
The sum over copies is $\sum_{a,b}\widehat{G}_{ab}^4= 4g^4+12h^4$ when both arguments lie on the same contour segment and
\begin{equation}
  \sum_{a,b}\widehat G_{ab}^4=g^4+18g^2h^2-24\epsilon gh^3+21h^4
  \label{eq:em-quartic}
\end{equation}
when they lie on different segments, with $\epsilon$ as in Eq.~\eqref{eq:em-reduced-self-energy}.
The real parts of $I$ evaluated at the saddles with $h=0$ and $h\neq0$ give $I_{\mathrm{sym}}$ and $I_{\mathrm{ssb}}$, respectively.
We also evaluate $I_1$ for a single copy without the insertion of $\mathcal R$, on the same contour and grid.
The particle propagator $f_1$, the Majorana propagator $g_1$, and the self-energy $\sigma_1$ satisfy
\begin{align}
  f_1           & =\left[I_L-\frac{8}{L^2}\Theta\sigma_1\right]^{-1}\Theta,
  \qquad g_1=f_1-f_1^{\top},\notag                                          \\
  \sigma_1(u,v) & =J^2\dot z(u)\dot z(v)g_1(u,v)^3.
  \label{eq:em-normalization}
\end{align}
We obtain $I_1$ by evaluating the action for a single copy at this solution and taking its real part.
These three actions give $m(\beta,t)=\min\{I_{\mathrm{sym}},I_{\mathrm{ssb}}\}-4I_1$.

We locate the phase boundary by solving $I_{\mathrm{ssb}}-I_{\mathrm{sym}}=0$ with Newton's method in $\beta$ or $t$.
At a saddle, $\delta I/\delta G=\delta I/\delta\Sigma=0$, so we calculate the derivative of the action with respect to $\beta$ or $t$ by keeping $G$ and $\Sigma$ fixed and differentiating only the factors $\dot z(u)\dot z(v)$.
We check these derivatives against finite differences of the actions, solving the SD equations separately at each nearby value of $\beta$ or $t$.
We obtain the phase diagram in Fig.~\ref{fig:syk-transition} by this procedure with $L=512$ integration points on the contour.

\textit{Convergence of the location of the phase boundary.}---We examine whether the computed location of the phase boundary converges as we increase $L$, and whether different methods of numerical integration give the same limit.
For each value of $L$, we first solve the SD equations to the prescribed tolerance and determine the location of the phase boundary.
In addition to the midpoint rule described above, we use Gauss-Legendre quadrature to check whether the phase boundary approaches the same limit with a different choice of integration points and weights.
We apply this quadrature separately on $[0,1]$ and $[1,2]$, placing $L/2$ integration points on each side of the discontinuity at $u=1$.
The integration points are nonuniformly spaced, and their weights vary from point to point~\cite{press2007numerical}.
Figure~\ref{fig:em-convergence} shows the critical time at $\beta=0$ and the critical inverse temperature at $t=0$, computed by both methods for several values of $L$.
In each case, the two methods give different results at finite $L$ but approach the same value as $L$ increases.
The approximately linear dependence of these critical values on $L^{-2}$ allows extrapolation to the continuum limit, giving $Jt_c=0.54249$ and $\beta_cJ=5.8784$.
We also find an approximately linear dependence on $L^{-2}$ for the critical time at $\beta J=4$, where both $\beta$ and $t$ are nonzero.

%%%=========== Supplemental Material ===========%%%
\onecolumngrid
\clearpage
\begin{center}
  \textbf{\large Supplemental Material: Universal Bound and Phase Transition in Many-Body Fermionic Non-Gaussianity}
\end{center}
\beginsupplement

\section{Haar moments and concentration of the MRE}
\label{sec:sm-haar-methods}

In this section, we derive the results on Haar-random states stated in the main text.
We express the first two moments of the purity as traces of Gaussian operators, evaluate them exactly, and then control the logarithm to obtain the mean and typical MRE.
Throughout this section, $N=2n\geq2$, $D=2^n$, $d=D/2$, and $\ket\psi$ is a normalized vector drawn from the complex Haar measure on the even-parity sector $\mathcal H_+$.
We write $\rho=\ket\psi\bra\psi$, $\bar\nu=\mathbb E_{\mathrm H}\nu_2$, and $\overline{M_2}=\mathbb E_{\mathrm H}[-\ln\nu_2]$.

\subsection{Haar moments as traces of Gaussian operators}

\begin{lemma}[Determinant formula for Haar moments]
  \label{lem:sm-haar-determinant-moments}
  Let $s\geq1$ be an integer and define
  \begin{equation}
    R=I_4-2\bm v\bm v^{\top},
    \qquad
    \bm v=\frac12(1,1,1,1)^{\top},
    \qquad R_s=R^{\oplus s}.
    \label{eq:sm-replica-reflection}
  \end{equation}
  Here, $R$ is the reflection $I_4-\bm1\bm1^{\top}/2$ used in the main text, and $R_s$ is the block-diagonal matrix with one copy of $R$ for each group of four copies.
  For $\boldsymbol\epsilon\in\{0,1\}^{4s}$, let $Q_{\boldsymbol\epsilon}=\operatorname{diag}((-1)^{\epsilon_1},(-1)^{\epsilon_2},\ldots,(-1)^{\epsilon_{4s}})$ change the sign on copy $a$ when $\epsilon_a=1$.
  For $\pi\in S_{4s}$, let $(P_\pi)_{ab}=\delta_{a,\pi(b)}$ be the matrix that sends copy $b$ to copy $\pi(b)$.
  Then
  \begin{equation}
    \mathbb E_{\mathrm H}\nu_2^s
    =\frac{1}{2^{4s}(d)_{4s}}
    \sum_{\pi\in S_{4s}}\sum_{\boldsymbol\epsilon\in\{0,1\}^{4s}}
    \det\!\left(I_{4s}+Q_{\boldsymbol\epsilon}R_sP_\pi\right)^n,
    \label{eq:sm-haar-determinant-moments}
  \end{equation}
  where $(d)_q=d(d+1)\cdots(d+q-1)$.
\end{lemma}

\begin{proof}
  We first evaluate the Haar mean of $\rho^{\otimes q}$.
  Every vector $\ket\psi^{\otimes q}$ lies in the symmetric subspace of $\mathcal H_+^{\otimes q}$, so the average of $\rho^{\otimes q}$ is supported on this subspace.
  The average also commutes with $W^{\otimes q}$ for every $W\in U(d)$ because the Haar measure is invariant under $\ket\psi\mapsto W\ket\psi$.
  The action of $U(d)$ on the symmetric subspace is irreducible, so Schur's lemma implies that the average is proportional to the projector onto that subspace.
  This projector is $(q!)^{-1}\sum_{\pi\in S_q}\mathsf P_\pi$ and has trace $\binom{d+q-1}{q}=(d)_q/q!$, where $\mathsf P_\pi$ sends tensor factor $a$ to position $\pi(a)$.
  Normalizing the average to have trace one gives
  \begin{equation}
    \mathbb E_{\mathrm H}\rho^{\otimes q}
    =\frac{1}{(d)_q}\sum_{\pi\in S_q}\mathsf P_\pi.
    \label{eq:sm-haar-state-moment}
  \end{equation}

  To evaluate the resulting traces, we associate a Gaussian unitary on the full Fock space with each matrix acting on the copy indices.
  Set $q=4s$ and let $\mathcal F_q$ be the full Fock space of $q$ copies.
  For a real orthogonal $q\times q$ matrix $W$, define the unitary $\widehat W$ on $\mathcal F_q$ by
  \begin{equation}
    \widehat W f_m^{(a)\dagger}\widehat W^\dagger
    =\sum_{b=1}^q W_{ba}f_m^{(b)\dagger},
    \qquad
    \widehat W\ket0=\ket0,
    \qquad m\in[n].
  \end{equation}
  These conditions determine its action on every occupation state, since such a state is a product of creation operators acting on $\ket0$.
  The conjugation action alone determines $\widehat W$ only up to an overall phase, which would change its trace, and the condition on the vacuum fixes this phase.
  These conditions also imply $\widehat W_1\widehat W_2=\widehat{W_1W_2}$, so multiplication of these operators corresponds to multiplication of their matrices.
  The trace of $\widehat W$ follows from summing over empty and occupied modes.
  For a fixed $m$, let $\widehat W_m$ be the operator induced by $W$ on the Fock space of the $q$ copies of mode $m$.
  We diagonalize $W$, with eigenvalues $\lambda_1,\lambda_2,\ldots,\lambda_q$, and use the corresponding linear combinations of these $q$ creation operators.
  The occupation state in which the modes in a subset $A\subseteq[q]$ are occupied then has eigenvalue $\prod_{a\in A}\lambda_a$ under $\widehat W_m$, and summing over all subsets gives
  \begin{equation}
    \Tr\widehat W_m
    =\sum_{A\subseteq[q]}\prod_{a\in A}\lambda_a
    =\prod_{a=1}^q(1+\lambda_a)
    =\det(I_q+W),
    \label{eq:sm-exterior-algebra-character}
  \end{equation}
  where the empty subset contributes one.
  Since $\widehat W$ acts separately and identically on each $m\in[n]$, the full trace is $\Tr_{\mathcal F_q}\widehat W=\det(I_q+W)^n$.

  We now show that $\mathcal R$ is of this form by identifying it with the parity operator of $n$ modes that are symmetric among the four copies.
  The fermionic SWAP $S_{13}$ exchanges copies 1 and 3, so it reverses the sign of the antisymmetric combinations $(\gamma_i^{(1)}-\gamma_i^{(3)})/\sqrt2$ while leaving the symmetric combinations and copies 2 and 4 unchanged.
  The parity operator of the antisymmetric modes acts in the same way on all Majorana operators and, like $S_{13}$, leaves the product vacuum invariant.
  Since these two properties determine the operator uniquely, $S_{13}$ is the parity operator of the antisymmetric modes.
  With $V=U_{12}U_{34}\Pi_3\Pi_4$, the definition in Eq.~\eqref{eq:purity-representations} gives $\mathcal R=V^\dagger S_{13}V$.
  The beam splitters first transform each antisymmetric combination into $(\gamma_i^{(1)}+\gamma_i^{(2)}-\gamma_i^{(3)}-\gamma_i^{(4)})/2$, and conjugation by $\Pi_3\Pi_4$ then reverses the last two signs, so that
  \begin{equation}
    \widetilde\gamma_i
    :=V^\dagger
    \frac{\gamma_i^{(1)}-\gamma_i^{(3)}}{\sqrt2}
    V
    =\frac12\sum_{a=1}^4\gamma_i^{(a)}.
  \end{equation}
  The corresponding complex modes and their parity operator are
  \begin{equation}
    \widetilde f_m
    =\frac{\widetilde\gamma_{2m-1}-i\widetilde\gamma_{2m}}2
    =\frac12\sum_{a=1}^4 f_m^{(a)},
    \qquad
    \mathcal R=\prod_{m=1}^n\left(I-2\widetilde f_m^\dagger\widetilde f_m\right).
  \end{equation}
  Conjugation by $\mathcal R$ reverses the sign of the symmetric combination $\widetilde\gamma_i=\sum_a v_a\gamma_i^{(a)}$ and leaves combinations orthogonal to $\bm v$ unchanged, so $\mathcal R\gamma_i^{(a)}\mathcal R^\dagger=\sum_bR_{ab}\gamma_i^{(b)}$.
  Each $\widetilde f_m$ annihilates the product vacuum $\ket0$, so $\mathcal R\ket0=\ket0$.
  These two properties match the definition of $\widehat W$, and hence $\mathcal R^{\otimes s}=\widehat R_s$.

  The permutations and the parity projectors are also of this form.
  For $W=P_\pi$, the operator $\widehat P_\pi$ sends each creation operator $f_m^{(a)\dagger}$ to $f_m^{(\pi(a))\dagger}$.
  Its action differs from an ordinary permutation of tensor factors by signs from reordering the creation operators.
  Exchanging two adjacent copies with $k_a$ and $k_b$ particles gives the sign $(-1)^{k_ak_b}$, which is $+1$ when both copies have even parity.
  Every permutation is a product of adjacent exchanges, so $\widehat P_\pi$ agrees with $\mathsf P_\pi$ on $\mathcal H_+^{\otimes q}$.
  Let $\mathcal P_+=\prod_{a=1}^q(I+\Pi_a)/2$ project onto $\mathcal H_+^{\otimes q}$.
  Using $\nu_2^s=\Tr(\rho^{\otimes q}\mathcal R^{\otimes s})$ and Eq.~\eqref{eq:sm-haar-state-moment}, we obtain
  \begin{equation}
    \mathbb E_{\mathrm H}\nu_2^s
    =\frac{1}{(d)_q}\sum_{\pi\in S_q}
    \Tr_{\mathcal F_q}\!\left(\mathcal P_+\mathcal R^{\otimes s}\widehat P_\pi\right).
    \label{eq:sm-haar-moment-fock-trace}
  \end{equation}
  The projector expands as $\mathcal P_+=2^{-q}\sum_{\boldsymbol\epsilon\in\{0,1\}^q}\prod_a\Pi_a^{\epsilon_a}$.
  Each operator $\prod_a\Pi_a^{\epsilon_a}$ reverses the sign of the creation operators on copy $a$ when $\epsilon_a=1$ and leaves the vacuum invariant, so it equals $\widehat Q_{\boldsymbol\epsilon}$.
  The product $\widehat Q_{\boldsymbol\epsilon}\widehat R_s\widehat P_\pi$ therefore corresponds to the matrix $Q_{\boldsymbol\epsilon}R_sP_\pi$, and Eq.~\eqref{eq:sm-exterior-algebra-character} gives
  \begin{equation}
    \Tr_{\mathcal H_+^{\otimes q}}
    \!\left(\mathcal R^{\otimes s}\mathsf P_\pi\right)
    =2^{-q}\sum_{\boldsymbol\epsilon\in\{0,1\}^q}
    \det\!\left(I_q+Q_{\boldsymbol\epsilon}R_sP_\pi\right)^n.
    \label{eq:sm-replicated-spinor-character}
  \end{equation}
  Substitution of this trace into Eq.~\eqref{eq:sm-haar-moment-fock-trace} proves Eq.~\eqref{eq:sm-haar-determinant-moments}.
\end{proof}

\subsection{Exact mean purity}

With $s=1$, Lemma~\ref{lem:sm-haar-determinant-moments} gives the mean purity as a sum over determinants of $4\times4$ matrices.
These matrices do not depend on $N$, which enters only through the exponent $n=N/2$ and the normalization $16(d)_4$.
The mean purity is therefore determined by the possible values of the determinants and the number of times each value occurs.

\begin{theorem}[Mean purity of Haar-random states]
  \label{thm:sm-haar-mean}
  For every even $N\geq2$, the mean purity of a Haar-random state in $\mathcal H_+$ is
  \begin{equation}
    \bar\nu
    =\frac{8\cdot3^N+12D^3+16D^2+48D+56}{D(D+2)(D+4)(D+6)}.
    \label{eq:sm-haar-exact-mean-sm}
  \end{equation}
  Its large-$N$ expansion is
  \begin{align}
    \bar\nu
     & =8\left(\frac34\right)^N
    \left[
      1+\frac32\left(\frac{2\sqrt2}{3}\right)^N
      +\mathcal{O}\left(2^{-N/2}\right)
      \right],
    \label{eq:sm-two-copy-purity-asymptotic} \\
    -\ln\bar\nu
     & =N\ln\frac43-\ln8
    -\frac32\left(\frac{2\sqrt2}{3}\right)^N
    +\mathcal{O}\left(\left(\frac89\right)^N\right).
    \label{eq:sm-two-copy-entropy-asymptotic}
  \end{align}
\end{theorem}

\begin{proof}
  For $s=1$, the sum in Eq.~\eqref{eq:sm-haar-determinant-moments} runs over $4!=24$ permutations $\pi$ and $2^4=16$ choices of $\boldsymbol\epsilon$, or $384$ pairs in total.
  Since $R=I_4-\bm1\bm1^\top/2$ and $\bm1^\top P_\pi=\bm1^\top$, each determinant can be computed from the integer matrix
  \begin{equation}
    A=2(I_4+Q_{\boldsymbol\epsilon}RP_\pi),
    \qquad
    A_{ij}=2\delta_{ij}+(-1)^{\epsilon_i}\left(2\delta_{i,\pi(j)}-1\right),
    \qquad
    \det(I_4+Q_{\boldsymbol\epsilon}RP_\pi)=\frac{\det A}{16}.
    \label{eq:sm-first-integer-determinants}
  \end{equation}
  Since $A$ has integer entries, expanding $\det A$ over the $24$ permutations of its columns gives its value exactly.

  To count how often each value occurs, we group the permutations $\pi$ by their cycle lengths.
  Relabeling the four copies leaves $R$ unchanged and permutes the $16$ sign choices, so permutations with the same cycle lengths give the same counts.
  It therefore suffices to choose one representative of each of the five types, evaluate its $16$ determinants, and multiply the counts by the number of permutations of that type.
  The following table lists, for each type, the number of pairs $(\pi,\boldsymbol\epsilon)$ giving each value of the determinant.
  \begin{equation}
    \begin{array}{c|r|rrrrrr}
      \text{Cycle lengths of }\pi & \text{Number of }\pi & \multicolumn{6}{c}{\det(I_4+Q_{\boldsymbol\epsilon}RP_\pi)}                           \\
                                  &                      & 9                                                           & 8  & 4  & 2  & 1  & 0   \\
      \hline
      (1,1,1,1)                   & 1                    & 0                                                           & 0  & 4  & 0  & 0  & 12  \\
      (2,1,1)                     & 6                    & 0                                                           & 6  & 0  & 24 & 0  & 66  \\
      (2,2)                       & 3                    & 0                                                           & 0  & 12 & 0  & 0  & 36  \\
      (3,1)                       & 8                    & 8                                                           & 0  & 0  & 0  & 56 & 64  \\
      (4)                         & 6                    & 0                                                           & 6  & 0  & 24 & 0  & 66  \\
      \hline
      \text{Total}                & 24                   & 8                                                           & 12 & 16 & 48 & 56 & 244
    \end{array}.
    \label{eq:sm-first-determinant-histogram}
  \end{equation}
  For example, $\pi=(123)$ and $\boldsymbol\epsilon=(1,1,1,0)$ give
  \begin{equation}
    A=\begin{pmatrix}
      3  & 1  & -1 & 1 \\
      -1 & 3  & 1  & 1 \\
      1  & -1 & 3  & 1 \\
      -1 & -1 & -1 & 3
    \end{pmatrix},
    \qquad
    \det A=144,
    \qquad
    \det(I_4+Q_{\boldsymbol\epsilon}RP_\pi)=9.
  \end{equation}
  Among the $16$ sign choices for this representative, only $(1,1,1,0)$ gives the value $9$.
  Since there are $\binom43\cdot2=8$ three-cycles, the value $9$ occurs eight times in total.
  In the sum in Eq.~\eqref{eq:sm-haar-determinant-moments}, each value enters as its $n$th power multiplied by its total count.
  The zero determinants do not contribute because $n\geq1$, so
  \begin{equation}
    \bar\nu
    =\frac{F_1}{16(d)_4},
    \qquad
    F_1=8\cdot9^n+12\cdot8^n+16\cdot4^n+48\cdot2^n+56.
    \label{eq:sm-first-moment-numerator}
  \end{equation}
  With $9^n=3^N$, $8^n=D^3$, $4^n=D^2$, $2^n=D$, and $16(d)_4=D(D+2)(D+4)(D+6)$, this is Eq.~\eqref{eq:sm-haar-exact-mean-sm}.

  For the large-$N$ expansion, we factor $8\cdot3^N$ out of the numerator and $D^4$ out of the denominator.
  The ratio of these factors is $8(3/4)^N$, and the term $12D^3$ in the numerator gives the relative correction $(3/2)(2\sqrt2/3)^N$.
  The other terms in the numerator give relative corrections of order $(2/3)^N$ or smaller, while the denominator divided by $D^4$ is $(1+2/D)(1+4/D)(1+6/D)=1+\mathcal{O}(2^{-N/2})$.
  These estimates prove Eq.~\eqref{eq:sm-two-copy-purity-asymptotic}.
  The expansion of $-\ln\bar\nu$ then follows from $-\ln(1+x)=-x+\mathcal{O}(x^2)$ applied to the relative correction.
  The square of the leading correction is of order $(8/9)^N$, while $2^{-N/2}$ decays faster, so the remainder is $\mathcal{O}((8/9)^N)$ as stated in Eq.~\eqref{eq:sm-two-copy-entropy-asymptotic}.
\end{proof}

\subsection{Exact second moment and concentration}

The second moment of the purity controls its fluctuations around the Haar mean.
Since the relative fluctuations decay exponentially in $N$, the MRE also concentrates, and its mean differs from $-\ln\bar\nu$ by an amount that vanishes as $N\to\infty$.

\begin{theorem}[Concentration of the purity and the MRE]
  \label{thm:sm-haar-concentration}
  For a Haar-random state in $\mathcal H_+$, the relative variance of the purity,
  \begin{equation}
    V_N=\frac{\operatorname{Var}_{\mathrm H}(\nu_2)}{\bar\nu^{\,2}}
    =\mathbb E_{\mathrm H}\!\left[\left(\frac{\nu_2}{\bar\nu}-1\right)^2\right],
    \label{eq:sm-relative-variance-definition}
  \end{equation}
  decays exponentially with $N$ as
  \begin{equation}
    V_N
    =144\left(\frac{25}{81}\right)^{N/2}
    +144\left(\frac{49}{162}\right)^{N/2}
    +\mathcal{O}\left(\left(\frac{200}{729}\right)^{N/2}\right).
    \label{eq:sm-relative-variance}
  \end{equation}
  Consequently, the MRE concentrates around $N\ln(4/3)-\ln8$.
  For every fixed $\epsilon>0$,
  \begin{equation}
    \lim_{N\to\infty}
    \Pr_{\mathrm H}\!\left(
    \left|M_2(\psi)-\left(N\ln\frac43-\ln8\right)\right|>\epsilon
    \right)=0,
    \label{eq:sm-haar-faf-mre-contrast}
  \end{equation}
  where $\Pr_{\mathrm H}$ denotes probability under the Haar measure.
  At every $N$, the Jensen gap between the mean MRE and $-\ln\bar\nu$ satisfies
  \begin{equation}
    0\leq \overline{M_2}+\ln\bar\nu
    \leq\left[1+4\ln\left(d\bar\nu\right)\right]V_N,
    \label{eq:sm-jensen-gap}
  \end{equation}
  so the mean MRE has the same asymptotic form,
  \begin{equation}
    \overline{M_2}=N\ln\frac43-\ln8+o(1)
    \qquad(N\to\infty).
    \label{eq:sm-haar-mean-and-typical-entropy}
  \end{equation}
  Equation~\eqref{eq:sm-haar-faf-mre-contrast} therefore also holds with $N\ln(4/3)-\ln8$ replaced by $\overline{M_2}$.
  In particular,
  \begin{equation}
    \lim_{N\to\infty}\frac{\overline{M_2}}{N}
    =\lim_{N\to\infty}\frac{-\ln\bar\nu}{N}
    =\ln\frac43.
    \label{eq:sm-two-copy-haar-density}
  \end{equation}
\end{theorem}

\begin{proof}
  We first compute the second moment by repeating the enumeration of determinants for eight copies.
  For $s=2$, we write Eq.~\eqref{eq:sm-haar-determinant-moments} as
  \begin{equation}
    \mathbb E_{\mathrm H}\nu_2^2=\frac{F_2}{2^8(d)_8}.
    \label{eq:sm-second-moment-definition}
  \end{equation}
  For each of the $8!$ permutations $\pi$ and $2^8$ sign choices $\boldsymbol\epsilon$, we form the integer matrix
  \begin{equation}
    B_{ij}=2\left[\delta_{ij}+(-1)^{\epsilon_i}(R\oplus R)_{i,\pi(j)}\right],
    \qquad 1\leq i,j\leq8.
    \label{eq:sm-integer-determinant-enumeration}
  \end{equation}
  Since $B=2(I_8+Q_{\boldsymbol\epsilon}R_2P_\pi)$, the determinant in the Haar moment is $\det B/2^8$.
  We evaluate $\det B$ by fraction-free Gaussian elimination in exact integer arithmetic, which avoids floating-point rounding errors~\cite{bareiss1968sylvester}.
  We start with $r=1$ and, at each step $k=1,2,\ldots,7$, look for a nonzero entry in column $k$ among rows $i\geq k$.
  If there is none, the determinant is zero.
  Otherwise, we exchange rows if necessary to place this entry at $B_{kk}$, set $p=B_{kk}$, and update
  \begin{equation}
    B_{ij}\leftarrow\frac{pB_{ij}-B_{ik}B_{kj}}r
    \quad(i,j>k),
    \qquad
    B_{ik}\leftarrow0\quad(i>k),
    \qquad r\leftarrow p,
    \label{eq:sm-fraction-free-elimination}
  \end{equation}
  where $r$ on the right-hand side is the pivot of the preceding step.
  Every division is exact, and after the seventh step $B_{88}$ equals $\det B$ up to the sign of the row exchanges.
  Collecting equal values of $\det B/2^8$, each raised to the power $n$ and weighted by its number of occurrences, we obtain
  \begin{align}
    F_2={} & 64\cdot81^n+192\cdot72^n+144\cdot64^n
    +2048\left(\frac{81}{2}\right)^n
    +6400\cdot36^n+4992\cdot32^n
    \notag                                         \\
           & +9216\cdot25^n
    +9216\left(\frac{49}{2}\right)^n
    +30464\cdot18^n+55360\cdot16^n
    \notag                                         \\
           & +36864\left(\frac{25}{2}\right)^n
    +109952\cdot9^n+135488\cdot8^n
    +231424\left(\frac92\right)^n
    \notag                                         \\
           & +457984\cdot4^n+563456\cdot2^n
    +978496+834560\cdot2^{-n}.
    \label{eq:sm-exact-second-moment-numerator}
  \end{align}
  For example, $\det B=20736$ occurs 64 times and contributes $64\cdot81^n$, while $\det B=6400$ occurs 9216 times and contributes $9216\cdot25^n$.
  The zero determinant occurs $6{,}855{,}600$ times and contributes nothing because $n\geq1$.
  Together with the coefficients in Eq.~\eqref{eq:sm-exact-second-moment-numerator}, this accounts for all $8!\cdot2^8=10{,}321{,}920$ pairs $(\pi,\boldsymbol\epsilon)$.

  The relative variance is the ratio of the second moment to the squared mean, minus one.
  With $F_1$ from Eq.~\eqref{eq:sm-first-moment-numerator}, it reads
  \begin{align}
    V_N & =\frac{(d)_4^2F_2}{(d)_8F_1^2}-1,
    \notag                                               \\
    \frac{(d)_8}{(d)_4^2}
        & =\frac{(d+4)(d+5)(d+6)(d+7)}{d(d+1)(d+2)(d+3)}
    =1+32\cdot2^{-n}+\mathcal{O}(4^{-n}).
    \label{eq:sm-moment-normalization}
  \end{align}
  The numerator of $V_N$ is thus proportional to $F_2-[(d)_8/(d)_4^2]F_1^2$.
  When the normalization factor and $F_1^2$ are expanded, the terms proportional to $81^n$, $72^n$, $64^n$, $(81/2)^n$, $36^n$, and $32^n$ cancel against those in $F_2$.
  For example, $F_1^2$ contains $256\cdot36^n$ and $192\cdot72^n$, and the product of the latter with the correction $32\cdot2^{-n}$ adds $6144\cdot36^n$, so the total matches the coefficient $6400$ in $F_2$.
  The remaining terms are
  \begin{equation}
    F_2-\frac{(d)_8}{(d)_4^2}F_1^2
    =9216\cdot25^n+9216\left(\frac{49}{2}\right)^n
    +\mathcal{O}\left(\left(\frac{81}{4}\right)^n\right).
    \label{eq:sm-variance-cancellation}
  \end{equation}
  Dividing by $[(d)_8/(d)_4^2]F_1^2=64\cdot81^n[1+\mathcal{O}((8/9)^n)]$ gives the two leading terms in Eq.~\eqref{eq:sm-relative-variance}.
  The relative correction in this denominator multiplies $(25/81)^n$ and produces an error of order $(200/729)^n$, which dominates the error of order $(1/4)^n$ from the remainder in Eq.~\eqref{eq:sm-variance-cancellation}.

  We next show that the MRE of most Haar-random states is close to $-\ln\bar\nu$.
  Set $X=\nu_2/\bar\nu$, so that $\mathbb E_{\mathrm H}X=1$ and $\mathbb E_{\mathrm H}(X-1)^2=V_N$.
  Since $M_2+\ln\bar\nu=-\ln X$, the event $|M_2+\ln\bar\nu|>\epsilon$ requires $X<e^{-\epsilon}$ or $X>e^\epsilon$.
  Either condition implies $|X-1|>1-e^{-\epsilon}$, and Chebyshev's inequality gives
  \begin{equation}
    \Pr_{\mathrm H}\!\left(|M_2+\ln\bar\nu|>\epsilon\right)
    \leq\frac{V_N}{(1-e^{-\epsilon})^2}.
    \label{eq:sm-mre-deviation-bound}
  \end{equation}
  For every fixed $\epsilon>0$, the right-hand side tends to zero because $V_N\to0$.
  This bound and the expansion of $-\ln\bar\nu$ in Eq.~\eqref{eq:sm-two-copy-entropy-asymptotic} prove Eq.~\eqref{eq:sm-haar-faf-mre-contrast}.

  To bound the Jensen gap, we must also control the logarithm for rare states with small purity.
  For this purpose, we use a lower bound on the purity that holds for every state.
  The balanced beam splitter maps the antisymmetric Majorana combinations to $-\gamma_i^{(2)}$, so $US_{12}U^\dagger=\Pi_2$.
  Since $S_{12}\ket\psi^{\otimes2}=\ket\psi^{\otimes2}$ for identical even inputs, the second output has even parity, and conservation of the total parity makes the first output even as well.
  The output $\mathcal C(\rho)$ is thus supported on the $d$-dimensional space $\mathcal H_+$, and
  \begin{equation}
    \nu_2\geq d^{-1}.
    \label{eq:sm-two-copy-rank-floor}
  \end{equation}
  The lower bound on the gap, $\overline{M_2}+\ln\bar\nu\geq0$, is Jensen's inequality for the convex function $-\ln X$.
  For the upper bound, we use $\mathbb E_{\mathrm H}(X-1)=0$ to write the gap as $\mathbb E_{\mathrm H}[-\ln X+X-1]$ and split the expectation into the events $X\geq1/2$ and $X<1/2$.
  On the first event,
  \begin{equation}
    -\ln X+X-1\leq(X-1)^2,
    \label{eq:sm-log-local-bound}
  \end{equation}
  because the difference of the two sides has derivative $(2X-1)(X-1)/X$ and hence attains its minimum over $X\geq1/2$ at $X=1$, where it vanishes.
  On the second event, Eq.~\eqref{eq:sm-two-copy-rank-floor} gives $-\ln X+X-1\leq-\ln X\leq\ln(d\bar\nu)$, and Chebyshev's inequality bounds the probability of this event by $4V_N$.
  Adding the two contributions gives
  \begin{equation}
    0\leq\overline{M_2}+\ln\bar\nu
    =\mathbb E_{\mathrm H}[-\ln X+X-1]
    \leq V_N+4V_N\ln(d\bar\nu),
    \label{eq:sm-jensen-gap-proof}
  \end{equation}
  which is Eq.~\eqref{eq:sm-jensen-gap}.
  Since $d^{-1}\leq\bar\nu\leq1$, we have $0\leq\ln(d\bar\nu)\leq\ln d=\mathcal{O}(N)$.
  The gap is therefore $\mathcal{O}(NV_N)$, which tends to zero because $V_N$ decays exponentially.
  Hence $\overline{M_2}=-\ln\bar\nu+o(1)$, and Eq.~\eqref{eq:sm-two-copy-entropy-asymptotic} gives Eq.~\eqref{eq:sm-haar-mean-and-typical-entropy}.
  Division of these expansions by $N$ proves Eq.~\eqref{eq:sm-two-copy-haar-density}.
\end{proof}

\section{Proof of the upper bound on the MRE}
\label{sec:sm-upper-bound-proof}

In this section, we prove the universal upper bound on the MRE stated in the main text, allowing inputs of either parity.
We decompose the beam splitter into $N$ commuting rotations, one for each Majorana index, and apply them one at a time.
Each rotation decreases the purity by at most a factor $3/4$, while the last rotation increases it by exactly a factor $4/3$.
The main step is an identity that controls the interference between the two parity blocks during this process.
Throughout this section, $N=2n\geq2$, $d=2^{n-1}$, $\rho=\ket\psi\bra\psi$, and $\Pi\ket\psi=p\ket\psi$ with $p\in\{+1,-1\}$ and $\langle\psi|\psi\rangle=1$.
We use $[N]=\{1,2,\ldots,N\}$ and $\|A\|_{\mathrm F}^2=\Tr(A^\dagger A)$.

\begin{theorem}[Universal upper bound on the MRE]
  \label{thm:sm-mre-upper-bound}
  For every even $N\geq2$ and every normalized pure state $\ket\psi$ of fixed parity,
  \begin{equation}
    \nu_2(\psi)\geq\frac{16}{9}\left(\frac34\right)^N,
    \qquad
    M_2(\psi)\leq N\ln\frac43-\ln\frac{16}{9}.
    \label{eq:sm-universal-mre-bound}
  \end{equation}
\end{theorem}

The proof uses two lemmas.
The first states that, for identical pure inputs, the complete convolution preserves the parity of each copy.
This is not true for a partial convolution, which can populate both local parity sectors.

\begin{lemma}[Parity of the convolution output]
  \label{lem:sm-output-parity}
  Let $U=\prod_{i=1}^N\exp[(\pi/8)\gamma_i^{(1)}\gamma_i^{(2)}]$ be the balanced beam splitter.
  Then
  \begin{equation}
    \Pi_1U\ket\psi^{\otimes2}
    =\Pi_2U\ket\psi^{\otimes2}
    =p\,U\ket\psi^{\otimes2}.
    \label{eq:sm-convolution-output-parity}
  \end{equation}
\end{lemma}

\begin{proof}
  The fermionic SWAP $S_{12}$ exchanges the Majorana operators of the two copies and leaves the tensor product of their Fock vacua invariant.
  It multiplies the antisymmetric combinations $(\gamma_i^{(1)}-\gamma_i^{(2)})/\sqrt2$ by $-1$ and leaves the symmetric combinations unchanged, so it is the parity operator of the antisymmetric combinations,
  \begin{equation}
    S_{12}
    =i^n\prod_{i=1}^N\frac{\gamma_i^{(1)}-\gamma_i^{(2)}}{\sqrt2}.
    \label{eq:sm-swap-antisymmetric-parity}
  \end{equation}
  The product is ordered by increasing $i$, and its phase agrees with that of $S_{12}$ because both operators leave the product vacuum invariant.
  The beam splitter acts as
  \begin{equation}
    U\gamma_i^{(1)}U^\dagger
    =\frac{\gamma_i^{(1)}-\gamma_i^{(2)}}{\sqrt2},
    \qquad
    U\gamma_i^{(2)}U^\dagger
    =\frac{\gamma_i^{(1)}+\gamma_i^{(2)}}{\sqrt2}.
    \label{eq:sm-balanced-majorana-action}
  \end{equation}
  Hence $U$ maps each antisymmetric combination to $-\gamma_i^{(2)}$, and since $N$ is even, $US_{12}U^\dagger=i^n\prod_i(-\gamma_i^{(2)})=\Pi_2$.
  For two identical states of parity $p$, the fermionic SWAP gives $S_{12}\ket\psi^{\otimes2}=p\ket\psi^{\otimes2}$.
  Therefore $\Pi_2U\ket\psi^{\otimes2}=US_{12}\ket\psi^{\otimes2}=pU\ket\psi^{\otimes2}$, which means that the second output has parity $p$.
  The unitary $U$ preserves the total parity $\Pi_1\Pi_2$, whose eigenvalue on the input is $p^2=1$, so the first output also has parity $p$.
\end{proof}

The second lemma controls the purity at intermediate stages, when only some of the rotations have been applied.
For $\Lambda\subseteq[N]$, we define the partial convolution and the purity of its reduced output by
\begin{equation}
  U_\Lambda
  =\prod_{i\in\Lambda}\exp\left(\frac\pi8\gamma_i^{(1)}\gamma_i^{(2)}\right),
  \qquad
  \nu_\Lambda
  =\Tr\left[\left(\Tr_2 U_\Lambda\rho^{\otimes2}U_\Lambda^\dagger\right)^2\right].
  \label{eq:sm-partial-convolution-definition}
\end{equation}
The factors commute, $\nu_\varnothing=1$, and $\nu_{[N]}=\nu_2(\psi)$.

\begin{lemma}[Purity under partial convolution]
  \label{lem:sm-partial-convolution}
  Let $\ket\psi$ be a normalized pure state of fixed parity and $\Lambda\subseteq[N]$.
  For every $j\notin\Lambda$,
  \begin{equation}
    \nu_{\Lambda\cup\{j\}}\geq\frac34\nu_\Lambda.
    \label{eq:sm-partial-purity-inequality}
  \end{equation}
  Moreover, when all but one of the factors have been applied, the exact relation
  \begin{equation}
    \nu_{[N]\setminus\{j\}}=\frac34\nu_{[N]}
    \label{eq:sm-terminal-identity}
  \end{equation}
  holds for every $j\in[N]$.
\end{lemma}

\begin{proof}
  We track the coefficient matrix of the two-copy pure state, whose singular values are its Schmidt coefficients.
  For this purpose, we represent the Majorana operators of the two copies on an ordinary tensor product as
  \begin{equation}
    \gamma_i^{(1)}=\gamma_i\otimes I,
    \qquad
    \gamma_i^{(2)}=\Pi\otimes\gamma_i.
    \label{eq:sm-tensor-copy-majoranas}
  \end{equation}
  The generators $\gamma_i^{(1)}\gamma_i^{(2)}$ commute and square to $-I$, so that
  \begin{equation}
    U_\Lambda=\prod_{i\in\Lambda}\left(\cos\frac\pi8+\sin\frac\pi8\gamma_i^{(1)}\gamma_i^{(2)}\right).
    \label{eq:sm-partial-beam-splitter}
  \end{equation}
  We fix the label $j\notin\Lambda$ of the next rotation, choose an orthonormal basis $\{\ket{e_\alpha}\}_{\alpha=1}^d$ of the sector with parity $p$, and set $\ket{o_\alpha}=\gamma_j\ket{e_\alpha}$, which gives an orthonormal basis of the opposite sector.
  Since the total parity remains even, the two copies always have the same parity, and the amplitudes of the two possible local parities can be collected into two matrices $E$ and $O$ as
  \begin{equation}
    U_\Lambda\ket\psi^{\otimes2}
    =\sum_{\alpha,\beta=1}^d E_{\alpha\beta}\ket{e_\alpha}_1\ket{e_\beta}_2
    +p\sum_{\alpha,\beta=1}^d O_{\alpha\beta}\ket{o_\alpha}_1\ket{o_\beta}_2.
  \end{equation}
  The row and column indices of $E$ and $O$ label basis states of copies 1 and 2, respectively.
  For even input, $E$ contains the even-even amplitudes and $O$ the odd-odd amplitudes.
  For odd input, their roles are reversed, and the factor $p$ fixes the relative sign so that the formulas below apply to either input parity.
  For $\Lambda=\varnothing$, we have $E=\psi\psi^\top$ and $O=0$, where $\psi$ is the coordinate vector of the input in the chosen basis.
  The reduced state of copy 1 is block diagonal with blocks $EE^\dagger$ and $OO^\dagger$, whose traces are the probabilities of the two local parities.
  Hence
  \begin{equation}
    \nu_\Lambda=\Tr[(EE^\dagger)^2]+\Tr[(OO^\dagger)^2].
  \end{equation}

  With this choice of basis, the next rotation acts directly on the two matrices.
  Indeed, $\gamma_j^{(1)}\gamma_j^{(2)}=(\gamma_j\Pi)\otimes\gamma_j$ sends $\ket{e_\alpha}_1\ket{e_\beta}_2$ to $p\ket{o_\alpha}_1\ket{o_\beta}_2$ and $\ket{o_\alpha}_1\ket{o_\beta}_2$ to $-p\ket{e_\alpha}_1\ket{e_\beta}_2$.
  The rotation labeled by $j$ therefore acts as
  \begin{equation}
    E'=\cos\frac\pi8 E-\sin\frac\pi8 O,
    \qquad
    O'=\sin\frac\pi8 E+\cos\frac\pi8 O.
    \label{eq:sm-block-update}
  \end{equation}
  This update rule already gives the terminal relation.
  Suppose that $j$ is the last label, so that $\Lambda=[N]\setminus\{j\}$.
  By Lemma~\ref{lem:sm-output-parity}, both copies have parity $p$ after the last rotation, so only the block with parity $p$ in each copy survives.
  Since $O$ denotes the block with parity opposite to that of the input, Eq.~\eqref{eq:sm-block-update} gives $O'=0$ after the last rotation and hence $O=-\tan(\pi/8)E$ immediately before it.
  The purities before and after the last rotation are therefore
  \begin{align}
    \nu_{[N]\setminus\{j\}}
              & =\bigl[1+\tan^4(\pi/8)\bigr]\Tr[(EE^\dagger)^2],
    \notag                                                       \\
    \nu_{[N]} & =\frac{\Tr[(EE^\dagger)^2]}{\cos^4(\pi/8)}.
    \label{eq:sm-terminal-block-purities}
  \end{align}
  Their ratio is $[1+\tan^4(\pi/8)]\cos^4(\pi/8)=\cos^4(\pi/8)+\sin^4(\pi/8)=3/4$, which proves Eq.~\eqref{eq:sm-terminal-identity}.

  For the inequality at a general step, we expand the purity after the rotation in terms of the two blocks of the reduced state and their interference term,
  \begin{equation}
    X=EE^\dagger,
    \qquad
    Y=OO^\dagger,
    \qquad
    Z=EO^\dagger+OE^\dagger.
    \label{eq:sm-block-quadratics}
  \end{equation}
  We expand $\Tr[(E'E'^\dagger)^2]+\Tr[(O'O'^\dagger)^2]$ using $\cos^4(\pi/8)+\sin^4(\pi/8)=3/4$, $\cos^2(\pi/8)\sin^2(\pi/8)=1/8$, and $\cos(\pi/8)\sin(\pi/8)[\cos^2(\pi/8)-\sin^2(\pi/8)]=1/4$, and obtain
  \begin{equation}
    \nu_{\Lambda\cup\{j\}}
    =\frac34\nu_\Lambda
    +\frac12\Tr(XY)
    +\frac14\Tr(Z^2)
    +\frac12\Tr[(Y-X)Z].
    \label{eq:sm-one-step-expansion}
  \end{equation}
  The first two correction terms are nonnegative because $X$ and $Y$ are positive semidefinite and $Z$ is Hermitian, but the last term is not manifestly nonnegative.
  We show below that it satisfies
  \begin{equation}
    \Tr[(Y-X)Z]=4\Tr(XY).
    \label{eq:sm-parity-block-identity}
  \end{equation}
  With this identity, Eq.~\eqref{eq:sm-one-step-expansion} expresses the change in purity under a single rotation as a sum of squared norms,
  \begin{equation}
    \nu_{\Lambda\cup\{j\}}-\frac34\nu_\Lambda
    =\frac52\norm{O^\dagger E}_{\mathrm F}^2
    +\frac14\norm{EO^\dagger+OE^\dagger}_{\mathrm F}^2
    \geq0,
    \label{eq:sm-recursive-purity-bound}
  \end{equation}
  which proves Eq.~\eqref{eq:sm-partial-purity-inequality}.

  To prove Eq.~\eqref{eq:sm-parity-block-identity}, we write both of its sides in terms of a single expectation value in the four-copy input state.
  On the ordinary tensor product of four Fock spaces, we represent the Majorana operators of the four copies as
  \begin{equation}
    \gamma_i^{(a)}
    =\Pi^{\otimes(a-1)}\otimes\gamma_i\otimes I^{\otimes(4-a)},
    \qquad a=1,2,3,4.
    \label{eq:sm-four-copy-tensor-majoranas}
  \end{equation}
  Let $\ket\Phi=U_\Lambda\ket\psi^{\otimes2}$ be the two-copy state before the rotation labeled by $j$.
  We write $\langle A\rangle=\bra\psi^{\otimes4}A\ket\psi^{\otimes4}$ for expectation values in the four-copy input and denote by $U_{\Lambda,ab}$ the partial convolution $U_\Lambda$ acting on copies $a$ and $b$.
  The two local parities in $\ket\Phi$ agree, so the operator $S_{13}\Pi_1$ acts on $\ket\Phi^{\otimes2}$ as the ordinary SWAP of copies 1 and 3.
  Here, the factor $\Pi_1$ cancels the sign produced by the fermionic SWAP, including the signs from moving past copy 2.
  Moreover, $(S_{13}\Pi_1)\ket\psi^{\otimes4}=\ket\psi^{\otimes4}$ for either input parity.
  The swap trick therefore gives
  \begin{equation}
    \nu_\Lambda=\langle\mathcal B_\Lambda\rangle,
    \qquad
    \mathcal B_\Lambda=(U_{\Lambda,12}U_{\Lambda,34})^\dagger S_{13}\Pi_1(U_{\Lambda,12}U_{\Lambda,34})(S_{13}\Pi_1)^\dagger.
    \label{eq:sm-partial-four-copy-purity}
  \end{equation}
  The last factor can be inserted because it leaves $\ket\psi^{\otimes4}$ invariant, and this is the step that uses the purity of the input.
  The expectation values that enter the proof are
  \begin{equation}
    \kappa_{ab}=\langle\mathcal B_\Lambda\gamma_j^{(a)}\gamma_j^{(b)}\rangle,
    \qquad 1\leq a<b\leq4,
  \end{equation}
  and we show that both sides of Eq.~\eqref{eq:sm-parity-block-identity} equal $-2\kappa_{13}$.

  The right-hand side is related to $\kappa_{13}$ directly, because $\kappa_{13}$ measures the overlap of the two parity blocks.
  The insertion $\gamma_j^{(1)}\gamma_j^{(3)}$ commutes with $S_{13}\Pi_1$ and, because $j\notin\Lambda$, with both partial convolutions.
  Undoing the conjugation in Eq.~\eqref{eq:sm-partial-four-copy-purity} gives
  \begin{equation}
    \kappa_{13}=\bra\Phi^{\otimes2}S_{13}\Pi_1\gamma_j^{(1)}\gamma_j^{(3)}\ket\Phi^{\otimes2}=-\Tr(\sigma\gamma_j\sigma\gamma_j),
    \label{eq:sm-partial-overlap-expectation}
  \end{equation}
  where $\sigma=\Tr_2\ket\Phi\bra\Phi$ is the reduced state of copy 1.
  To check the sign, consider a component of $\ket\Phi^{\otimes2}$ in which copies 1 and 2 have parity $(-1)^r$ and copies 3 and 4 have parity $(-1)^s$, with $r,s\in\{0,1\}$.
  The insertion flips the parities of copies 1 and 3, after which $S_{13}\Pi_1$ contributes the sign $(-1)^{r+s}$.
  A nonzero overlap with $\bra\Phi^{\otimes2}$ requires $r\neq s$, which gives the minus sign in Eq.~\eqref{eq:sm-partial-overlap-expectation}.
  In the bases chosen above,
  \begin{equation}
    \sigma=\begin{pmatrix}X&0\\0&Y\end{pmatrix},
    \qquad
    \gamma_j=\begin{pmatrix}0&I\\I&0\end{pmatrix},
    \qquad
    \Tr(\sigma\gamma_j\sigma\gamma_j)=2\Tr(XY),
  \end{equation}
  and thus $\kappa_{13}=-2\Tr(XY)$.

  The left-hand side requires the explicit form of $\mathcal B_\Lambda$.
  The operator $S_{13}\Pi_1$ leaves the product vacuum invariant and has the factorization
  \begin{equation}
    S_{13}\Pi_1
    =\prod_{i=1}^N\exp\left(\frac\pi4\gamma_i^{(1)}\gamma_i^{(3)}\right)
    =\prod_{i=1}^N\frac{1+\gamma_i^{(1)}\gamma_i^{(3)}}{\sqrt2},
  \end{equation}
  where the two sides act identically on all Majorana operators and leave the product vacuum invariant, which fixes their relative phase.
  Conjugation by $(U_{\Lambda,12}U_{\Lambda,34})^\dagger$ replaces $\gamma_i^{(1)}$ and $\gamma_i^{(3)}$ by $(\gamma_i^{(1)}+\gamma_i^{(2)})/\sqrt2$ and $(\gamma_i^{(3)}+\gamma_i^{(4)})/\sqrt2$ when $i\in\Lambda$.
  The factors with $i\notin\Lambda$ cancel against the last factor in Eq.~\eqref{eq:sm-partial-four-copy-purity}, leaving
  \begin{equation}
    \mathcal B_\Lambda=\prod_{i\in\Lambda}V_i,
    \qquad
    V_i=\frac12\left[1+\frac{(\gamma_i^{(1)}+\gamma_i^{(2)})(\gamma_i^{(3)}+\gamma_i^{(4)})}{2}\right]
    (1-\gamma_i^{(1)}\gamma_i^{(3)}).
    \label{eq:sm-partial-four-copy-factors}
  \end{equation}
  By the Clifford relations, each factor expands as
  \begin{equation}
    4V_i=3+\gamma_i^{(1)}\gamma_i^{(2)}\gamma_i^{(3)}\gamma_i^{(4)}-\gamma_i^{(1)}\gamma_i^{(2)}-\gamma_i^{(1)}\gamma_i^{(3)}+\gamma_i^{(1)}\gamma_i^{(4)}+\gamma_i^{(2)}\gamma_i^{(3)}+\gamma_i^{(2)}\gamma_i^{(4)}-\gamma_i^{(3)}\gamma_i^{(4)}.
    \label{eq:sm-partial-four-copy-expansion}
  \end{equation}
  To isolate the interference term, we replace the angle $\pi/8$ of the rotation labeled by $j$ with a variable $\theta$.
  Let $\nu(\theta)$ be the reduced purity of $\exp(\theta\gamma_j^{(1)}\gamma_j^{(2)})\ket\Phi$.
  In the two-copy description, Eq.~\eqref{eq:sm-block-update} with $\pi/8$ replaced by $\theta$ shows that the two reduced blocks have derivatives $-Z$ and $Z$ at $\theta=0$, so
  \begin{equation}
    \nu'(0)=2\Tr[(Y-X)Z].
    \label{eq:sm-partial-purity-derivative}
  \end{equation}
  In the four-copy description, the same rotation multiplies $\mathcal B_\Lambda$ by one more factor, $\nu(\theta)=\langle\mathcal B_\Lambda V_j(\theta)\rangle$, where
  \begin{equation}
    V_j(\theta)=\frac12\bigl[1+(\cos2\theta\,\gamma_j^{(1)}+\sin2\theta\,\gamma_j^{(2)})(\cos2\theta\,\gamma_j^{(3)}+\sin2\theta\,\gamma_j^{(4)})\bigr](1-\gamma_j^{(1)}\gamma_j^{(3)}).
  \end{equation}
  By the Clifford relations, its derivative at $\theta=0$ is
  \begin{equation}
    V_j'(0)=\gamma_j^{(1)}\gamma_j^{(4)}+\gamma_j^{(2)}\gamma_j^{(3)}
    -\gamma_j^{(1)}\gamma_j^{(2)}-\gamma_j^{(3)}\gamma_j^{(4)}.
  \end{equation}
  Comparing the two expressions for $\nu'(0)$ gives
  \begin{equation}
    \Tr[(Y-X)Z]=\frac12\left(\kappa_{14}+\kappa_{23}-\kappa_{12}-\kappa_{34}\right).
    \label{eq:sm-partial-interference-expectation}
  \end{equation}

  It remains to reduce this combination to $\kappa_{13}$, for which we use two symmetries.
  Both $\mathcal B_\Lambda$ and the input are invariant under the signed permutations of the copies
  \begin{align}
    P:\quad(\gamma_i^{(1)},\gamma_i^{(2)},\gamma_i^{(3)},\gamma_i^{(4)})
     & \longmapsto(\gamma_i^{(2)},-\gamma_i^{(3)},-\gamma_i^{(1)},\gamma_i^{(4)}),
    \notag                                                                         \\
    Q:\quad(\gamma_i^{(1)},\gamma_i^{(2)},\gamma_i^{(3)},\gamma_i^{(4)})
     & \longmapsto(\gamma_i^{(2)},-\gamma_i^{(1)},-\gamma_i^{(4)},\gamma_i^{(3)}),
    \label{eq:sm-partial-replica-symmetries}
  \end{align}
  applied simultaneously to every Majorana index $i$.
  Substitution in Eq.~\eqref{eq:sm-partial-four-copy-expansion} shows that each transformation preserves $V_i$ and hence $\mathcal B_\Lambda$.
  Each is implemented by an even permutation of the copies together with parity operators on two copies, so it also leaves $\ket\psi^{\otimes4}$ invariant, including when $p=-1$.
  Invariance under $P$ gives $\kappa_{13}=\kappa_{12}=-\kappa_{23}$ and $\kappa_{14}=\kappa_{24}=-\kappa_{34}$, while invariance under $Q$ gives $\kappa_{14}=\kappa_{23}$.
  Together, these relations imply
  \begin{equation}
    \kappa_{12}=\kappa_{13}=\kappa_{34}
    =-\kappa_{14}=-\kappa_{23}=-\kappa_{24}.
    \label{eq:sm-partial-bilinear-relations}
  \end{equation}
  Each of the four terms in Eq.~\eqref{eq:sm-partial-interference-expectation} therefore equals $-\kappa_{13}$, and $\Tr[(Y-X)Z]=-2\kappa_{13}=4\Tr(XY)$.
  This proves Eq.~\eqref{eq:sm-parity-block-identity} and completes the proof.
\end{proof}

\begin{proof}[Proof of Theorem~\ref{thm:sm-mre-upper-bound}]
  We choose any ordering of the $N$ Majorana indices and denote the last one by $j$.
  Starting from $\nu_\varnothing=1$, we apply Eq.~\eqref{eq:sm-partial-purity-inequality} to the first $N-1$ rotations and then use Eq.~\eqref{eq:sm-terminal-identity} to obtain
  \begin{equation}
    \nu_2(\psi)
    =\nu_{[N]}
    =\frac43\nu_{[N]\setminus\{j\}}
    \geq\frac43\left(\frac34\right)^{N-1}
    =\frac{16}{9}\left(\frac34\right)^N.
    \label{eq:sm-terminal-purity-bound}
  \end{equation}
  The bound on $M_2$ follows by taking the negative logarithm.
\end{proof}

Together with the Haar mean, this bound determines the maximal MRE density in the large-$N$ limit, although it need not be saturated at finite $N$.

\begin{corollary}[Maximal MRE density]
  \label{cor:sm-maximal-mre-density}
  Let $M_{2,\max}(N)$ be the maximum of $M_2$ over normalized pure states in the even-parity sector.
  Then
  \begin{equation}
    \overline{M_2}\leq M_{2,\max}(N)
    \leq N\ln\frac43-\ln\frac{16}{9},
    \qquad
    \lim_{N\to\infty}\frac{M_{2,\max}(N)}{N}=\ln\frac43.
    \label{eq:sm-maximal-mre-density}
  \end{equation}
\end{corollary}

\begin{proof}
  The maximum is at least the Haar mean, and Theorem~\ref{thm:sm-mre-upper-bound} supplies the upper bound.
  Theorem~\ref{thm:sm-haar-concentration} gives $\overline{M_2}=N\ln(4/3)-\ln8+o(1)$.
  The limit follows by dividing the two bounds by $N$ and taking $N\to\infty$.
\end{proof}

\section{Real spinors and a construction above the Haar mean}
\label{sec:sm-above-haar-construction}

In this section, we prove Theorem~\ref{thm:real-spinors-above-haar} of the main text.
We first define the real structure and compute the mean purity of uniformly distributed real spinors exactly.
For infinitely many $N$, this mean lies below the Haar mean, and we turn the resulting averaging argument into a deterministic construction of states whose MRE provably exceeds the Haar mean at finite $N$.
Throughout this section, $N=2n\geq8$ is divisible by eight, $D=2^n$, and $d=D/2$.
We use $f_m=(\gamma_{2m-1}-i\gamma_{2m})/2$ and the occupation basis $\ket{x_1x_2\cdots x_n}=(f_1^\dagger)^{x_1}(f_2^\dagger)^{x_2}\cdots(f_n^\dagger)^{x_n}\ket{0^n}$.

\subsection{The real structure and an explicit basis}

\begin{lemma}[Real structure of the even-parity spinors]
  \label{lem:sm-real-structure}
  Let $K$ be complex conjugation in the occupation basis and define
  \begin{equation}
    T=\gamma_1\gamma_3\cdots\gamma_{2n-1}K.
    \label{eq:sm-real-spinor-antiunitary}
  \end{equation}
  Then $T$ is antiunitary, $T^2=1$, $T\Pi T^{-1}=\Pi$, and $T\gamma_jT^{-1}=-\gamma_j$ for every $j$.
  Its fixed space in $\mathcal H_+$ has real dimension $d$ and admits a basis that is orthonormal as a complex basis of $\mathcal H_+$.
\end{lemma}

\begin{proof}
  In the occupation basis, $\gamma_{2m-1}$ has real matrix elements and $\gamma_{2m}$ has imaginary ones.
  The product $\gamma_1\gamma_3\cdots\gamma_{2n-1}$ is therefore real and commutes with $K$, so $T^2=(\gamma_1\gamma_3\cdots\gamma_{2n-1})^2=(-1)^{n(n-1)/2}=1$ because $n$ is divisible by four.
  Since $n$ is even, conjugation by this product changes the sign of each odd-indexed Majorana operator and leaves each even-indexed one unchanged, while $K$ does the opposite.
  Thus $T\gamma_jT^{-1}=-\gamma_j$ for every $j$.
  Moreover, $i^n$ is real because $n$ is even, and the $N$ sign changes multiply to $(-1)^N=1$, so $T\Pi T^{-1}=\Pi$.

  For the deterministic construction below, we also need an explicit ordered basis of the fixed space.
  The operator $T$ maps each occupation string to its complement, so we choose one representative from each complementary pair.
  For each occupation string $x$ with an even number of occupied modes and $x_1=0$, let $\bar x$ be its bitwise complement and set $s_x=(-1)^{\sum_{m=1}^n(n-m)x_m}$.
  Applying the odd-indexed Majorana operators from right to left gives $T\ket x=s_x\ket{\bar x}$, and $T^2=1$ then gives $T\ket{\bar x}=s_x\ket x$.
  Each pair therefore contains the two $T$-fixed vectors
  \begin{equation}
    \ket{e_{x,+}}=\frac{\ket x+s_x\ket{\bar x}}{\sqrt2},
    \qquad
    \ket{e_{x,-}}=\frac{i(\ket x-s_x\ket{\bar x})}{\sqrt2},
    \label{eq:sm-explicit-real-spinor-basis}
  \end{equation}
  which are orthonormal over the complex numbers.
  The complementary pairs exhaust the even-parity sector, so these vectors form an orthonormal basis of $\mathcal H_+$, and the fixed space has real dimension $d$.
  We order the strings $x$ lexicographically and place $e_{x,+}$ before $e_{x,-}$.
\end{proof}

As in the main text, we call a $T$-fixed vector a real spinor, although its coefficients in the occupation basis need not be real.
We denote the uniform average over the unit sphere of the fixed space by $\mathbb E_{\mathbb R}$ and the mean purity of real spinors by $\bar\nu_{\mathbb R}=\mathbb E_{\mathbb R}\nu_2$.

\subsection{Mean purity of uniformly sampled real spinors}

The determinant formula in Lemma~\ref{lem:sm-haar-determinant-moments} relies on the unitary invariance of the Haar measure and does not apply directly to real spinors.
We instead expand the four-copy operator $\mathcal R$ in products of Majorana operators, which reduces the average to moments of real quadratic forms on the sphere.
For $A=\{i_1<i_2<\cdots<i_\ell\}\subseteq[N]$, we write $\gamma_A=\gamma_{i_1}\gamma_{i_2}\cdots\gamma_{i_\ell}$ and $\gamma_\varnothing=I$.
A coloring $c:[N]\to\{1,2,3,4\}$ defines four blocks $A_a=c^{-1}(a)$ and the inversion number $\operatorname{inv}(c)=|\{(i,j):i<j,\ c(i)>c(j)\}|$.

\begin{lemma}[Majorana expansion of the purity]
  \label{lem:sm-real-correlation-expansion}
  For a normalized $\ket\psi\in\mathcal H_+$,
  \begin{equation}
    \nu_2(\psi)
    =2^{-N}
    \sum_{\substack{c:[N]\to\{1,2,3,4\}\\ |A_a|\ {\rm even}\ \text{for all }a}}
    (-1)^{\operatorname{inv}(c)}
    \prod_{a=1}^4\bra\psi\gamma_{A_a}\ket\psi.
    \label{eq:sm-real-correlation-expansion}
  \end{equation}
\end{lemma}

\begin{proof}
  As shown in the proof of Lemma~\ref{lem:sm-haar-determinant-moments}, $\mathcal R$ is the parity operator of the symmetric combinations $\sum_av_a\gamma_i^{(a)}$ with $\bm v$ in Eq.~\eqref{eq:sm-replica-reflection}.
  Its action on the Majorana operators, together with the requirement that it leave the Fock vacuum invariant, fixes the representation
  \begin{equation}
    \mathcal R
    =\frac{i^n}{2^N}
    \prod_{i=1}^N
    \left(\gamma_i^{(1)}+\gamma_i^{(2)}
    +\gamma_i^{(3)}+\gamma_i^{(4)}\right),
    \label{eq:sm-real-reflection-product}
  \end{equation}
  with the product ordered by increasing $i$.
  The normalized combinations $(\gamma_i^{(1)}+\gamma_i^{(2)}+\gamma_i^{(3)}+\gamma_i^{(4)})/2$ satisfy the Majorana relations, so their parity operator reverses them and leaves the orthogonal combinations unchanged.
  The complex fermions formed from these combinations annihilate the product vacuum, which fixes the overall phase to agree with that of $\mathcal R$.

  Each term in the expansion of the product assigns every Majorana index to one of the four copies and thus defines a coloring $c$.
  The term has the coefficient $2^{-N}i^n$, and sorting its Majorana operators by copy gives the sign $(-1)^{\operatorname{inv}(c)}$.
  Terms with a block of odd size vanish, because $\gamma_A$ with odd $|A|$ has zero expectation value in a parity eigenstate.
  When all blocks are even, the operators on different copies have even parity, so the expectation value factorizes without further signs from fermionic ordering.
  Since $i^n=1$ for $N\equiv0\pmod8$, this proves Eq.~\eqref{eq:sm-real-correlation-expansion}.
\end{proof}

\begin{lemma}[Moments of commuting quadratic forms on the sphere]
  \label{lem:sm-sphere-moments}
  Let $j\geq1$ be an integer and let $u$ be uniformly distributed on the unit sphere of $\mathbb R^d$.
  Let $M_1,M_2,\ldots,M_j$ be commuting real symmetric $d\times d$ matrices.
  Suppose that $M_1M_2\cdots M_j=\eta I_d$ with $\eta\in\{+1,-1\}$, and that $\Tr\prod_{a\in S}M_a=0$ for every nonempty proper subset $S\subsetneq[j]$.
  Then
  \begin{equation}
    \mathbb E_u\prod_{a=1}^j(u^\top M_au)
    =\eta\,\frac{2^{j-1}(j-1)!}{(d+2)(d+4)\cdots(d+2j-2)},
    \label{eq:sm-real-sphere-single-cycle}
  \end{equation}
  where the denominator is the empty product, equal to one, when $j=1$.
\end{lemma}

\begin{proof}
  We separate the length and direction of a Gaussian vector.
  Let $g\in\mathbb R^d$ have independent standard Gaussian coordinates and write $g=ru$, where $u$ is uniform on the unit sphere and independent of $r$.
  Since $\mathbb E r^{2j}=d(d+2)\cdots(d+2j-2)$,
  \begin{equation}
    \mathbb E_u\prod_{a=1}^j(u^\top M_au)
    =\frac{\mathbb E_g\prod_{a=1}^j(g^\top M_ag)}{d(d+2)\cdots(d+2j-2)}.
    \label{eq:sm-real-sphere-gaussian-moment}
  \end{equation}
  We evaluate the numerator by Wick's theorem, representing each quadratic form by a vertex with two coordinate factors.
  Because each vertex has two factors, a pairing decomposes into cycles, and summing the coordinate indices around a cycle gives the trace of the product of its matrices.
  Commutativity makes each trace depend only on the subset of matrices in the cycle.
  For example, the three pairings for $j=2$ give
  \begin{equation*}
    \mathbb E_g[(g^\top M_1g)(g^\top M_2g)]
    =\Tr M_1\,\Tr M_2+2\Tr(M_1M_2).
  \end{equation*}
  The first term represents two separate cycles, and the second represents the two pairings that join the vertices.
  If a pairing has more than one cycle, each cycle involves a nonempty proper subset of the matrices, and the pairing vanishes by assumption.
  For $j\geq3$, a single cycle has $(j-1)!/2$ unoriented cyclic orders and two choices at each vertex for which coordinate factor is paired with each neighbor, giving $2^{j-1}(j-1)!$ contractions.
  For $j=1$ there is one such pairing, and for $j=2$ there are two, in agreement with the same expression.
  Each contraction with a single cycle contributes $\Tr(\eta I_d)=\eta d$, and dividing by the denominator in Eq.~\eqref{eq:sm-real-sphere-gaussian-moment} gives Eq.~\eqref{eq:sm-real-sphere-single-cycle}.
\end{proof}

\begin{lemma}[Exact mean for real spinors]
  \label{lem:sm-real-spinor-mean}
  Let $N\geq8$ be divisible by eight, and sample a normalized $T$-fixed vector uniformly in $\mathcal H_+$.
  Its mean purity is
  \begin{equation}
    \bar\nu_{\mathbb R}
    =\frac{96\Re[(2+i)^N]+12D^3+76D^2+192D+8\cdot3^N+344}
    {D(D+4)(D+8)(D+12)}.
    \label{eq:sm-real-mean-purity}
  \end{equation}
  With $\alpha=\arctan(1/2)$, the relative difference satisfies
  \begin{equation}
    \frac{\bar\nu-\bar\nu_{\mathbb R}}{\bar\nu}
    =-12\left(\frac{\sqrt5}{3}\right)^N\cos(N\alpha)
    +\mathcal{O}\!\left(2^{-N/2}\right).
    \label{eq:sm-real-complex-mean-difference-sm}
  \end{equation}
\end{lemma}

\begin{proof}
  We write a real spinor as $\ket\psi=\sum_{a=1}^du_a\ket{e_a}$ in the basis of Eq.~\eqref{eq:sm-explicit-real-spinor-basis}.
  Since the vectors $\ket{e_a}$ are $T$-fixed and orthonormal over the complex numbers, $\ket\psi$ is a normalized real spinor precisely when $u$ is a real unit vector, and $\mathbb E_{\mathbb R}$ is the uniform average $\mathbb E_u$ over the unit sphere of $\mathbb R^d$.
  For even $|A|$, the operator $\gamma_A$ preserves $\mathcal H_+$ and commutes with $T$, so its matrix $M_A$ in this basis is real and $\bra\psi\gamma_A\ket\psi=u^\top M_Au$.
  Since $M_A$ is real, its transpose equals its adjoint, $M_A^\top=(-1)^{|A|(|A|-1)/2}M_A$.
  The quadratic forms of a real antisymmetric matrix vanish, so only colorings whose blocks all have sizes divisible by four contribute to Eq.~\eqref{eq:sm-real-correlation-expansion}.

  Consider such a coloring $c$ with $j$ nonempty blocks $A_{a_1},A_{a_2},\ldots,A_{a_j}$, where $a_1<a_2<\cdots<a_j$ label the copies that carry them.
  We apply Lemma~\ref{lem:sm-sphere-moments} to the matrices $M_{A_{a_1}},M_{A_{a_2}},\ldots,M_{A_{a_j}}$.
  They are real symmetric, and they commute because the blocks are disjoint and have even sizes.
  Their ordered product is the matrix of $\gamma_{A_{a_1}}\gamma_{A_{a_2}}\cdots\gamma_{A_{a_j}}=(-1)^{\operatorname{inv}(c)}\gamma_{[N]}$.
  Since $i^n=1$ for $N\equiv0\pmod8$, we have $\gamma_{[N]}=\Pi$, which acts as the identity on $\mathcal H_+$, so $\eta=(-1)^{\operatorname{inv}(c)}$.
  The product over a nonempty proper subset of the blocks is $\pm\gamma_A$, where $A$ is their union and satisfies $A\neq\varnothing$ and $A\neq[N]$.
  Up to this sign, its trace on $\mathcal H_+$ is $\tfrac12\Tr[(I+\Pi)\gamma_A]$, and both terms vanish.
  For $i\in A$, conjugation gives $\gamma_i\gamma_A\gamma_i=-\gamma_A$ because $A$ is nonempty and even, so $\Tr\gamma_A=0$.
  The operator $\Pi\gamma_A$ is proportional to $\gamma_{[N]\setminus A}$, and applying the same argument to the nonempty even complement $[N]\setminus A$ gives $\Tr(\Pi\gamma_A)=0$.
  Lemma~\ref{lem:sm-sphere-moments} therefore applies, and its sign $\eta=(-1)^{\operatorname{inv}(c)}$ cancels the coloring sign in Eq.~\eqref{eq:sm-real-correlation-expansion}.

  We now sum over colorings.
  Let $C_j$ be the number of ordered partitions of $[N]$ into $j$ nonempty blocks whose sizes are divisible by four.
  A contributing coloring with $j$ nonempty blocks is specified by such a partition and by the $j$ copies that carry the blocks, which can be chosen in $\binom4j$ ways.
  Together with the factor $2^{j-1}(j-1)!$ in Eq.~\eqref{eq:sm-real-sphere-single-cycle}, this gives the coefficients $4$, $6\cdot2=12$, $4\cdot8=32$, and $1\cdot48=48$ for $j=1,2,3,4$, respectively, and hence
  \begin{equation}
    \bar\nu_{\mathbb R}
    =2^{-N}\left[
      4+\frac{12C_2}{d+2}
      +\frac{32C_3}{(d+2)(d+4)}
      +\frac{48C_4}{(d+2)(d+4)(d+6)}
      \right].
    \label{eq:sm-real-spinor-partition-mean}
  \end{equation}
  We count $C_j$ by filtering the block sizes with fourth roots of unity.
  For an integer $m\geq0$, $\frac14\sum_{s=0}^3i^{sm}$ equals one when $m$ is divisible by four and zero otherwise.
  Applying this filter to all assignments of the $N$ indices to labeled blocks and using inclusion-exclusion to remove empty blocks gives
  \begin{equation}
    C_j
    =\sum_{\ell=1}^j(-1)^{j-\ell}\binom j\ell
    \frac1{4^\ell}\sum_{s_1,s_2,\ldots,s_\ell=0}^3
    \left(\sum_{a=1}^\ell i^{s_a}\right)^N.
    \label{eq:sm-real-spinor-partition-filter}
  \end{equation}
  We evaluate the filter using $(1+i)^N=D$ and $(3+i)^N=D(2-i)^N$, both valid for $N\equiv0\pmod8$, and obtain
  \begin{align}
    C_1 & =1,
    \qquad C_2=\frac{D^2}4+\frac D2-2,
    \notag                                                                           \\
    C_3 & =\frac38\Re[(2+i)^N]+\frac{3^N}{16}-\frac{3D^2}4-\frac{3D}2+\frac{57}{16},
    \notag                                                                           \\
    C_4 & =\frac{D^4}{64}+\frac{3D^3}{32}+\frac{7D^2}4+\frac{27D}8
    -\frac{3^N}4+\left(\frac D8-\frac32\right)\Re[(2+i)^N]-\frac{25}4.
    \label{eq:sm-real-spinor-partition-counts}
  \end{align}
  With these counts and $D=2d$, Eq.~\eqref{eq:sm-real-spinor-partition-mean} becomes Eq.~\eqref{eq:sm-real-mean-purity}.

  For the asymptotic form, we use $\Re[(2+i)^N]=5^{N/2}\cos(N\alpha)$.
  The numerators in Eqs.~\eqref{eq:sm-haar-exact-mean-sm} and \eqref{eq:sm-real-mean-purity} share the terms $8\cdot3^N+12D^3$.
  The remaining polynomial terms give relative corrections of order $(2/3)^N$, while the two denominators differ by a relative amount of order $1/D=2^{-N/2}$.
  The ratio of the two means is therefore
  \begin{equation*}
    \frac{\bar\nu_{\mathbb R}}{\bar\nu}
    =1+12\left(\frac{\sqrt5}{3}\right)^N\cos(N\alpha)
    \left[1+\mathcal{O}\!\left(\left(\frac{2\sqrt2}{3}\right)^N\right)\right]
    +\mathcal{O}\!\left(2^{-N/2}\right),
  \end{equation*}
  where the correction in the brackets comes from $12D^3$ in the common part of the numerators.
  The product of the oscillating term with this correction is of order $(2\sqrt{10}/9)^N$, and $2\sqrt{10}/9<1/\sqrt2$, so it is included in $\mathcal{O}(2^{-N/2})$.
  This proves Eq.~\eqref{eq:sm-real-complex-mean-difference-sm}.
\end{proof}

\subsection{A deterministic construction of states whose MRE exceeds the Haar mean}

When $\bar\nu_{\mathbb R}<\bar\nu$, some real spinor has purity below $\bar\nu$ and hence MRE above $-\ln\bar\nu$.
To exceed the Haar mean $\overline{M_2}$, the gain in MRE above $-\ln\bar\nu$ must also exceed the Jensen gap $\overline{M_2}+\ln\bar\nu$.
We therefore choose a target purity and compare the corresponding gain in MRE with an explicit bound on this gap.

For $N\in8\mathbb N$, we define $\delta_N$ as half the relative difference between the two mean purities and choose the target purity $b_N$ halfway between them,
\begin{equation}
  \delta_N=\frac{\bar\nu-\bar\nu_{\mathbb R}}{2\bar\nu},
  \qquad
  b_N=\bar\nu(1-\delta_N)=\frac{\bar\nu+\bar\nu_{\mathbb R}}2.
  \label{eq:sm-above-haar-threshold}
\end{equation}
By Eq.~\eqref{eq:sm-real-complex-mean-difference-sm}, the leading term of $\delta_N$ is proportional to $-\cos(N\alpha)$ and oscillates with $N$.
When $\delta_N>0$, the target lies strictly between the two means, $\bar\nu_{\mathbb R}<b_N<\bar\nu$, and since the mean purity of real spinors is below $b_N$, some real spinor has purity below this target.

For the gap, we use $\ln(d\bar\nu)\leq d\bar\nu-1$ in Theorem~\ref{thm:sm-haar-concentration} and obtain
\begin{equation}
  0\leq\overline{M_2}+\ln\bar\nu\leq B_N,
  \qquad
  B_N=(4d\bar\nu-3)V_N,
  \label{eq:sm-rational-haar-jensen-bound}
\end{equation}
where both $\bar\nu$ and $B_N$ are rational numbers determined by the exact Haar moments.
For any state with $\nu_2(\psi)<b_N$, Eqs.~\eqref{eq:sm-above-haar-threshold} and \eqref{eq:sm-rational-haar-jensen-bound} give
\begin{equation}
  M_2(\psi)-\overline{M_2}
  >-\ln(1-\delta_N)-B_N
  \geq\delta_N-B_N,
  \label{eq:sm-above-haar-certificate}
\end{equation}
where the last step uses $-\ln(1-x)\geq x$ for $x<1$, which applies because $\bar\nu_{\mathbb R}>0$ implies $\delta_N<1/2$.
Thus a state whose purity is below the target has MRE above the Haar mean whenever $\delta_N>B_N$, in which case the margin $\delta_N-B_N$ is positive.
We carry out the construction at the system sizes satisfying this condition,
\begin{equation}
  \mathcal I=\{N\in8\mathbb N:\delta_N>B_N\}.
  \label{eq:sm-above-haar-size-set}
\end{equation}

\begin{theorem}[A deterministic family with MRE above the Haar mean]
  \label{thm:sm-deterministic-above-haar}
  For every $N\in\mathcal I$, let $e_1,e_2,\ldots,e_d$ be the orthonormal basis of $T$-fixed vectors in Eq.~\eqref{eq:sm-explicit-real-spinor-basis}, in the order specified there.
  A finite deterministic procedure produces integers $q_1,q_2,\ldots,q_d\in\{-4,-3,\ldots,4\}$, not all zero, such that the state
  \begin{equation}
    \ket{\psi_N}
    =\frac{\sum_{a=1}^d q_a\ket{e_a}}{\sqrt{\sum_{a=1}^d q_a^2}}
    \label{eq:sm-integer-real-spinor-family}
  \end{equation}
  has $\nu_2(\psi_N)<b_N$ and
  \begin{equation}
    M_2(\psi_N)-\overline{M_2}>\delta_N-B_N>0.
    \label{eq:sm-deterministic-quenched-advantage}
  \end{equation}
  The set $\mathcal I$ is infinite.
  Applying any Gaussian unitary to $\psi_N$ gives a state that obeys the same inequality.
\end{theorem}

\begin{proof}
  We encode the target bound on the purity in a polynomial $P_N$ and show that its expectation is negative for Gaussian coordinates.
  We then reproduce this expectation exactly with a distribution on finitely many integers and fix the coordinates one at a time by the method of conditional expectations~\cite{alon2016probabilistic}.

  We first express the desired bound on the purity as the condition that a polynomial is negative.
  For $x\in\mathbb R^d$, set $\ket{v(x)}=\sum_a x_a\ket{e_a}$ and define
  \begin{equation}
    P_N(x)=\bra{v(x)}^{\otimes4}\mathcal R\ket{v(x)}^{\otimes4}-b_N\left(\sum_a x_a^2\right)^4.
    \label{eq:sm-above-haar-polynomial}
  \end{equation}
  Equation~\eqref{eq:purity-representations} gives $P_N(x)=\|x\|^8[\nu_2(v(x)/\|x\|)-b_N]$ for $x\ne0$, and $P_N(0)=0$.
  By construction, $P_N$ is homogeneous of degree eight, and $P_N(x)<0$ is precisely the desired inequality for the purity of the normalized state.

  The coefficients of $P_N$ are rational, so every subsequent comparison can be made exactly.
  The matrix elements of the Majorana operators in the occupation basis and the coefficients in Eq.~\eqref{eq:sm-real-reflection-product} lie in $\mathbb Q(i)$, while the eight factors of $1/\sqrt2$ from the basis vectors multiply to $1/16$.
  The coefficients of the first term are therefore in $\mathbb Q(i)$ and are real because this term equals a purity times $\|x\|^8$ for every real $x$.
  Finally, $b_N$ is rational by the exact formulas for the two mean purities.

  Lemma~\ref{lem:sm-real-spinor-mean} implies that $P_N$ has negative expectation for a standard real Gaussian vector $g\in\mathbb R^d$.
  The direction $u=g/\|g\|$ is uniform on the unit sphere and independent of $\|g\|$, so $\ket{v(u)}$ is a uniformly distributed normalized real spinor and $\mathbb E P_N(u)=\bar\nu_{\mathbb R}-b_N$.
  By homogeneity and $\mathbb E\|g\|^8=d(d+2)(d+4)(d+6)$, we obtain
  \begin{equation*}
    \mathbb E P_N(g)
    =\mathbb E\|g\|^8\,\mathbb E P_N(u)
    =d(d+2)(d+4)(d+6)(\bar\nu_{\mathbb R}-b_N)<0.
  \end{equation*}
  The inequality is strict because $\bar\nu_{\mathbb R}-b_N=-\bar\nu\delta_N$ and $N\in\mathcal I$ implies $\delta_N>B_N\geq0$.

  To obtain integer coordinates, we replace each Gaussian coordinate by a discrete random variable with the same moments through degree eight.
  We choose a symmetric random variable $Z$ on the nine integers from $-4$ to $4$ and write $p_0=\Pr(Z=0)$ and $p_j=\Pr(Z=j)=\Pr(Z=-j)$ for $1\leq j\leq4$.
  Normalization and the second, fourth, sixth, and eighth moments of a standard real Gaussian variable give the five linear equations
  \begin{equation*}
    \begin{aligned}
      p_0+2(p_1+p_2+p_3+p_4)         & =1,   \\
      2(p_1+4p_2+9p_3+16p_4)         & =1,   \\
      2(p_1+16p_2+81p_3+256p_4)      & =3,   \\
      2(p_1+64p_2+729p_3+4096p_4)    & =15,  \\
      2(p_1+256p_2+6561p_3+65536p_4) & =105.
    \end{aligned}
  \end{equation*}
  The factor of two accounts for the two signs, and the odd moments vanish automatically by symmetry.
  Solving this system gives the positive probabilities
  \begin{equation}
    \renewcommand{\arraystretch}{1.8}
    \begin{array}{c|ccccc}
      z        & 0                & \pm1             & \pm2          & \pm3 & \pm4 \\
      \hline
      \Pr(Z=z) & \dfrac{115}{288} & \dfrac{29}{120}
               & \dfrac{13}{240}  & \dfrac{11}{2520} & \dfrac1{6720}
    \end{array},
    \label{eq:sm-discrete-real-moments}
  \end{equation}

  These moments suffice to preserve the expectation of $P_N$ when we replace $g$ by independent random variables $Z_1,Z_2,\ldots,Z_d$ distributed as $Z$.
  Independence makes the expectation of each monomial a product of moments of the individual coordinates.
  Since $P_N$ has total degree eight, every exponent in these monomials is at most eight, so all the required moments agree with those of the Gaussian coordinates.
  Consequently,
  \begin{equation}
    \mathbb E P_N(Z_1,Z_2,\ldots,Z_d)
    =\mathbb E P_N(g)<0.
    \label{eq:sm-negative-polynomial-mean}
  \end{equation}

  We now turn the negative average in Eq.~\eqref{eq:sm-negative-polynomial-mean} into a deterministic choice of the coordinates, taken in the order specified in Eq.~\eqref{eq:sm-explicit-real-spinor-basis}.
  At step $a$, the coordinates $q_1,q_2,\ldots,q_{a-1}$ have already been fixed.
  We average over the remaining coordinates $Z_{a+1},Z_{a+2},\ldots,Z_d$ and choose $q_a$ to minimize
  \begin{equation}
    \mathbb E\!\left[P_N(q_1,q_2,\ldots,q_{a-1},q,Z_{a+1},Z_{a+2},\ldots,Z_d)\right]
    \quad\text{over }q\in\{-4,-3,\ldots,4\},
    \label{eq:sm-conditional-expectation-construction}
  \end{equation}
  breaking ties by choosing the smallest integer.
  All nine conditional expectations are rational and are obtained by substituting the known moments into a finite polynomial.

  This choice keeps the conditional expectation negative at every step.
  Before $q_a$ is fixed, the conditional expectation is the weighted average of the nine candidate values, with the probabilities in Eq.~\eqref{eq:sm-discrete-real-moments}.
  Their minimum is no larger than this average, and the average is negative by induction from Eq.~\eqref{eq:sm-negative-polynomial-mean}, so
  \begin{equation*}
    \mathbb E\!\left[P_N(q_1,q_2,\ldots,q_a,Z_{a+1},Z_{a+2},\ldots,Z_d)\right]
    \leq\mathbb E\!\left[P_N(q_1,q_2,\ldots,q_{a-1},Z_a,Z_{a+1},\ldots,Z_d)\right]<0.
  \end{equation*}
  After the last step, no coordinates remain random, and the conditional expectation is simply $P_N(q_1,q_2,\ldots,q_d)<0$.
  Since $P_N(0)=0$, the chosen vector is nonzero, and normalizing it gives $\nu_2(\psi_N)<b_N$.
  Equation~\eqref{eq:sm-above-haar-certificate} with $N\in\mathcal I$ then gives Eq.~\eqref{eq:sm-deterministic-quenched-advantage}.

  It remains to prove that $\delta_N>B_N$ for infinitely many $N\in8\mathbb N$.
  By Eq.~\eqref{eq:sm-two-copy-purity-asymptotic}, $4d\bar\nu=16\left(\frac{3}{2\sqrt2}\right)^N[1+o(1)]$ grows with $N$, and Eq.~\eqref{eq:sm-relative-variance} gives $V_N=144(5/9)^N[1+o(1)]$.
  Together with Eq.~\eqref{eq:sm-real-complex-mean-difference-sm}, we obtain
  \begin{equation}
    B_N=2304\left(\frac{25}{72}\right)^{N/2}[1+o(1)],
    \qquad
    \delta_N=-6\left(\frac{\sqrt5}{3}\right)^N\cos(N\alpha)
    +\mathcal{O}\!\left(2^{-N/2}\right).
    \label{eq:sm-above-haar-certificate-asymptotic}
  \end{equation}
  If $\alpha/\pi$ were rational, $e^{2i\alpha}$ and its inverse would be roots of unity, so their sum $2\cos(2\alpha)=6/5$ would be an algebraic integer.
  A rational algebraic integer must be an integer, so $\alpha/\pi$ is irrational.
  Hence the sequence $8\ell\alpha$ is dense modulo $2\pi$, and infinitely many $N=8\ell$ satisfy $\cos(N\alpha)\leq-1/2$.
  For these sizes, $\delta_N\geq(3+o(1))(\sqrt5/3)^N$, while both $B_N$ and the remainder in the expansion of $\delta_N$ decay faster.
  Every sufficiently large size in this sequence therefore belongs to $\mathcal I$, so $\mathcal I$ is infinite.
  The final assertion follows from the invariance of the MRE under Gaussian unitaries.
\end{proof}

The procedure requires only finitely many operations on rational numbers, although we do not have a polynomial bound on its running time or on the cost of preparing $\ket{\psi_N}$.
The guaranteed excess of the MRE over the Haar mean is exponentially small in $N$, consistent with the fact that typical Haar-random states and states with maximal MRE have the same MRE density in the large-$N$ limit.

\section{Comparison with measures based on the covariance matrix}
\label{sec:sm-covariance-comparison}

In this section, we give the calculations behind the comparison between the MRE and the FAF in the main text.
We first show that real spinors have vanishing covariance matrices.
We then compute the covariance matrix and the purity for superpositions of the empty and fully occupied Fock states, a family that includes both the cat states and the state $\ket\eta$ of the main text.
These results give the two examples in the main text, namely states with identical covariance matrices but extensively different MRE, and states with equal MRE but different FAF.
Throughout this section, $N=2n$, $\Gamma_{ij}=(i/2)\bra\psi[\gamma_i,\gamma_j]\ket\psi$ is the covariance matrix, and $\mathcal F_k(\psi)=n-\tfrac12\Tr[(\Gamma^{\top}\Gamma)^k]$ is the FAF.

We first take $N\in8\mathbb N$, as in Sec.~\ref{sec:sm-above-haar-construction}, and show that the antiunitary symmetry of a real spinor forces its covariance matrix to vanish.
For $i\neq j$, the operator $i\gamma_i\gamma_j$ is Hermitian, and Lemma~\ref{lem:sm-real-structure} together with $TiT^{-1}=-i$ gives $T(i\gamma_i\gamma_j)T^{-1}=-i\gamma_i\gamma_j$.
For an antiunitary $T$ with $T\ket\psi=\ket\psi$, every operator $X$ satisfies $\bra\psi X\ket\psi=\bra\psi TXT^{-1}\ket\psi^*$.
Applied to $X=i\gamma_i\gamma_j$, this gives $\Gamma_{ij}=-\Gamma_{ij}^*=-\Gamma_{ij}$, because $\Gamma_{ij}$ is real, so $\Gamma_{ij}=0$.
The diagonal entries vanish by definition.
The states $\psi_N$ of Theorem~\ref{thm:sm-deterministic-above-haar} therefore have $\Gamma=0$ and $\mathcal F_k(\psi_N)=N/2$ for every integer $k\geq1$.

We next consider the states
\begin{equation}
  \ket\phi=\sqrt p\,\ket{0^n}+\sqrt q\,\ket{1^n},
  \qquad p,q\geq0,\quad p+q=1,
  \label{eq:sm-empty-full-superposition}
\end{equation}
with even $n\geq4$, so that $\ket{1^n}$ has even parity.
The cat state $\ket{\mathrm{cat}_n}$ corresponds to $p=q=1/2$.
In each Fock state, a bilinear of Majorana operators from two different modes has zero expectation value, and the occupation bilinear $i\gamma_{2m-1}\gamma_{2m}=1-2f_m^\dagger f_m$ has expectation value $+1$ in $\ket{0^n}$ and $-1$ in $\ket{1^n}$.
The cross terms between the two Fock states vanish, because a bilinear changes at most two occupations and cannot connect $\ket{0^n}$ and $\ket{1^n}$ for $n\geq4$.
The only possibly nonzero entries of the covariance matrix are therefore $\Gamma_{2m-1,2m}=-\Gamma_{2m,2m-1}=p-q$, and all $N$ singular values of $\Gamma$ equal $|p-q|$.
In particular, $\Gamma=0$ and $\mathcal F_k(\mathrm{cat}_n)=n$ for every $k\geq1$.

To compute the purity, we track the occupations after the beam splitter.
By Eq.~\eqref{eq:sm-balanced-majorana-action}, $U$ maps $f_m^{(1)\dagger}$ and $f_m^{(2)\dagger}$ to $(f_m^{(1)\dagger}-f_m^{(2)\dagger})/\sqrt2$ and $(f_m^{(1)\dagger}+f_m^{(2)\dagger})/\sqrt2$, respectively.
It leaves $\ket{0^n}\ket{0^n}$ invariant, and it also leaves $\ket{1^n}\ket{1^n}$ invariant because each product $f_m^{(1)\dagger}f_m^{(2)\dagger}$ is unchanged.
Applied to a fully occupied copy and an empty copy, $U$ produces a superposition of all ways of distributing the $n$ occupied modes between the two copies, each with an amplitude of magnitude $2^{-n/2}$.
For $S\subseteq[n]$, let $\ket{1_S}$ denote the occupation state with precisely the modes in $S$ filled.
The components $\ket{1_S}\ket{1_{S^c}}$ obtained from $\ket{1^n}\ket{0^n}$ and $\ket{0^n}\ket{1^n}$ differ by the factor $(-1)^{|S^c|}=(-1)^{|S|}$, since $n$ is even, so they cancel for odd $|S|$ and add for even $|S|$.
The two-copy output is therefore
\begin{equation}
  U\ket\phi^{\otimes2}
  =p\ket{0^n}\ket{0^n}
  +2\sqrt{pq}\,2^{-n/2}\sum_{\substack{S\subseteq[n]\\|S|\ {\rm even}}}\epsilon_S\ket{1_S}\ket{1_{S^c}}
  +q\ket{1^n}\ket{1^n},
  \qquad \epsilon_S\in\{\pm1\},
  \label{eq:sm-cat-output}
\end{equation}
where the signs $\epsilon_S$ come from fermionic ordering and satisfy $\epsilon_\varnothing=\epsilon_{[n]}=1$, since these two components require no reordering.
After tracing out the second copy, each occupation state with an even number of occupied modes receives the diagonal weight $r=4pq\,2^{-n}$ from the middle sum.
The empty and fully occupied states receive the additional weights $p^2$ and $q^2$, respectively.
Their coherence comes from the terms with $S=\varnothing$ and $S=[n]$, and since $\epsilon_\varnothing=\epsilon_{[n]}=1$, the corresponding off-diagonal element of the reduced state is $p\sqrt r+q\sqrt r=\sqrt r$.
Summing the squares of the diagonal weights and the contribution $2r$ of the two off-diagonal elements gives
\begin{equation}
  \nu_2(\phi)
  =(p^2+r)^2+(q^2+r)^2+(2^{n-1}-2)r^2+2r,
  \qquad r=4pq\,2^{-n}.
  \label{eq:sm-superposition-purity}
\end{equation}

For the cat state, $p=q=1/2$ and $r=2^{-n}$, so
\begin{equation}
  \nu_2(\mathrm{cat}_n)
  =2\left(\frac14+2^{-n}\right)^2+(2^{n-1}-2)2^{-2n}+2^{1-n}
  =\frac18+\frac7{2^{n+1}}.
  \label{eq:sm-cat-purity}
\end{equation}
Hence $M_2(\mathrm{cat}_n)\to\ln8$ as $n\to\infty$.
In contrast, $M_2(\psi_N)$ exceeds the Haar mean and obeys the universal upper bound, so Theorems~\ref{thm:sm-haar-concentration} and \ref{thm:sm-mre-upper-bound} give $M_2(\psi_N)=N\ln(4/3)+\mathcal{O}(1)$ along the constructed sequence.
Since $n=N/2$ is even and at least four for $N\in8\mathbb N$, comparing $\psi_N$ with $\mathrm{cat}_{N/2}$ proves Eqs.~\eqref{eq:equal-faf} and \eqref{eq:equal-faf-mre-gap}.

Conversely, equal MRE does not imply equal FAF.
On six complex fermion modes, we consider the states of the main text,
\begin{equation}
  \ket\chi=\ket{\mathrm{cat}_4}\otimes\ket{00},
  \qquad \ket\eta=\sqrt{\frac{1+z}{2}}\ket{0^6}+\sqrt{\frac{1-z}{2}}\ket{1^6},
  \qquad z=\sqrt{\frac{2\sqrt{65}}5-3}.
\end{equation}
The Gaussian factor $\ket{00}$ does not change the MRE by additivity, so Eq.~\eqref{eq:sm-cat-purity} with $n=4$ gives $\nu_2(\chi)=11/32$.
The state $\ket\eta$ has the form of Eq.~\eqref{eq:sm-empty-full-superposition} with $n=6$ and $p-q=z$, and Eq.~\eqref{eq:sm-superposition-purity} gives
\begin{equation}
  \nu_2(\eta)=\frac{23+90z^2+15z^4}{128}=\frac{11}{32},
  \qquad M_2(\chi)=M_2(\eta)=\ln\frac{32}{11}.
\end{equation}
The chosen value of $z$ satisfies $5z^4+30z^2-7=0$, which gives $\nu_2(\eta)=11/32$.
For the FAF, the covariance matrix of $\chi$ is the direct sum of the vanishing covariance matrix of $\ket{\mathrm{cat}_4}$ and that of the two-mode vacuum, so it has eight zero singular values and four unit singular values.
All twelve singular values of the covariance matrix of $\eta$ equal $z$.
Therefore $\mathcal F_k(\chi)=4$ and $\mathcal F_k(\eta)=6(1-z^{2k})>4$ for every $k\geq1$, where the last inequality follows from $0<z^2<1/3$.

\section{Higher-order MRE: Haar means and upper bounds}
\label{sec:sm-higher-mre}

In this section, we define higher-order MREs through convolutions of $k\geq3$ copies and determine their Haar means and maximal densities.
The two-copy output has fixed parity, and its typical purity is parametrically larger than the inverse dimension of that parity sector.
In contrast, the typical $k$-copy output approaches the maximally mixed state on the full Fock space.
Throughout this section, $N=2n$, $D=2^n$, and $d=D/2$ refer to a single copy, and all limits are taken at fixed $k\geq3$.

The $k$-copy convolution mixes the copies by a Gaussian unitary and retains the equal-weight combination of their modes.
We choose a real orthogonal $k\times k$ matrix $O_k$ whose first row is $(1,1,\ldots,1)/\sqrt k$, and let $U_k$ be the Gaussian unitary that leaves the Fock vacuum invariant and acts as $U_k^\dagger f_m^{(a)}U_k=\sum_b(O_k)_{ab}f_m^{(b)}$.
For $\rho=\ket\psi\bra\psi$, we define
\begin{align}
  \mathcal C_k(\rho) & =\Tr_{2,3,\ldots,k}(U_k\rho^{\otimes k}U_k^\dagger), \\
  \nu_k(\psi)        & =\Tr[\mathcal C_k(\rho)^2],
  \qquad M_k(\psi)=-\frac{\ln\nu_k(\psi)}{k-1}.
  \label{eq:sm-higher-mre-definition}
\end{align}
The other rows of $O_k$ only change the basis of the discarded modes and do not affect $\mathcal C_k$.
For $k=2$, Eq.~\eqref{eq:sm-balanced-majorana-action} gives $U^\dagger\gamma_i^{(1)}U=(\gamma_i^{(1)}+\gamma_i^{(2)})/\sqrt2$ and $U^\dagger\gamma_i^{(2)}U=(\gamma_i^{(2)}-\gamma_i^{(1)})/\sqrt2$.
Since $U$ also leaves the Fock vacuum invariant, the choice of rows $(1,1)/\sqrt2$ and $(-1,1)/\sqrt2$ gives $U_2=U$, so $M_2$ is the MRE of the main text.

\begin{theorem}[Haar mean and maximal density for $k\geq3$]
  \label{thm:sm-general-k-haar}
  For a Haar-random state in $\mathcal H_+$,
  \begin{equation}
    \mu_k:=\mathbb E_{\mathrm H}\nu_k
    =D^{-1}\left[1+k^2r_k^n+\mathcal{O}(q_k^n)\right],
    \label{eq:sm-general-k-asymptotic}
  \end{equation}
  where
  \begin{equation}
    r_k=\frac{(k+1)^2}{2k^2},
    \qquad q_k=\max\left\{\frac12,\frac{(k+2)^2}{4k^2}\right\}<r_k<1.
  \end{equation}
  Every pure state of fixed parity obeys $M_k\leq\ln D/(k-1)$, and
  \begin{equation}
    \lim_{N\to\infty}\frac{\mathbb E_{\mathrm H}M_k}{N}
    =\lim_{N\to\infty}\frac{\max_{\psi\in\mathcal H_+}M_k(\psi)}{N}
    =\frac{\ln2}{2(k-1)}.
    \label{eq:sm-general-k-density}
  \end{equation}
  Moreover, $\|\mathcal C_k(\rho)-I/D\|_1\to0$ in probability under the Haar measure.
\end{theorem}

\subsection{The purity as a Gaussian expectation value}

As for $k=2$, we write the purity as an expectation value in copies of the input state, now with $2k$ copies.
We label these copies by $1,2,\ldots,2k$ and group them into two batches, $A=\{1,2,\ldots,k\}$ and $B=\{k+1,k+2,\ldots,2k\}$, each of which is convolved by $U_k$.
The modes retained from the two batches correspond to the unit vectors
\begin{equation}
  \bm h_A=\frac1{\sqrt k}(1,1,\ldots,1,0,0,\ldots,0)^{\top},
  \qquad
  \bm h_B=\frac1{\sqrt k}(0,0,\ldots,0,1,1,\ldots,1)^{\top}
  \label{eq:sm-batch-vectors}
\end{equation}
in the space of copy labels.
Unlike the two-copy output, $\mathcal C_k(\rho)$ need not be supported in a single parity sector, so the swap trick requires a parity correction.
Let $S_{\mathrm o}$ and $S_{\mathrm f}$ be the ordinary and fermionic SWAPs of the two retained outputs, and let $\Pi_A$ and $\Pi_B$ be their parities.
The two SWAPs differ by a minus sign precisely when both outputs have odd parity, so
\begin{equation}
  S_{\mathrm o}=\tfrac12S_{\mathrm f}(I+\Pi_A+\Pi_B-\Pi_A\Pi_B).
  \label{eq:sm-ordinary-graded-swap}
\end{equation}
Each of the four terms is a Gaussian operator that acts only on the two retained combinations.
After conjugation by the convolutions, they correspond to the following orthogonal matrices on the copy labels.
The fermionic SWAP exchanges $\bm h_A$ and $\bm h_B$ and therefore corresponds to the reflection $R_-$.
The term with both parities also reverses the signs of the exchanged combinations and corresponds to the reflection $R_+$.
These reflections are
\begin{equation}
  R_\pm=I_{2k}-2\bm v_\pm\bm v_\pm^{\top},
  \qquad
  \bm v_\pm=\frac{\bm h_A\pm\bm h_B}{\sqrt2}.
  \label{eq:sm-two-batch-reflections}
\end{equation}
The two terms with a single parity give the rotations by $\pm\pi/2$ in the plane spanned by $\bm h_A$ and $\bm h_B$, that is, $J_k$ and $J_k^{-1}$ with
\begin{equation}
  J_k
  =I_{2k}-\bm h_A\bm h_A^{\top}-\bm h_B\bm h_B^{\top}
  +\bm h_A\bm h_B^{\top}-\bm h_B\bm h_A^{\top}.
  \label{eq:sm-Jk-definition}
\end{equation}
As in Sec.~\ref{sec:sm-haar-methods}, we write $\widehat W$ for the Gaussian unitary that implements a real orthogonal matrix $W$ on the copy labels and leaves the Fock vacuum invariant.

\begin{proposition}[Purity as a Gaussian expectation value]
  \label{prop:sm-parity-corrected-observable}
  For every pure state $\ket\psi$ of fixed parity, the purity is
  \begin{equation}
    \nu_k(\psi)=\Tr\left(\rho^{\otimes2k}\widehat J_k\right).
    \label{eq:sm-Jk-purity-observable}
  \end{equation}
\end{proposition}

\begin{proof}
  Applying the swap trick with Eq.~\eqref{eq:sm-ordinary-graded-swap} to the two retained outputs and conjugating by the convolutions gives
  \begin{equation}
    \nu_k(\psi)
    =\frac12\Tr\left[\rho^{\otimes2k}
      \left(\widehat R_-+\widehat J_k+\widehat J_k^{-1}-\widehat R_+\right)\right].
    \label{eq:sm-four-gaussian-swap-resolution}
  \end{equation}
  Two symmetries of $\rho^{\otimes2k}$ simplify this expression.
  Applying the parity operator to every input in batch $B$ changes the expectation value only by the factor $p^{2k}=1$, where $p$ is the parity of $\ket\psi$, and it conjugates $R_-$ to $R_+$.
  The two reflection terms in Eq.~\eqref{eq:sm-four-gaussian-swap-resolution} therefore cancel.
  Exchanging the two identical batches conjugates $J_k$ to $J_k^{-1}$, so the remaining two terms are equal, and Eq.~\eqref{eq:sm-four-gaussian-swap-resolution} reduces to Eq.~\eqref{eq:sm-Jk-purity-observable}.
\end{proof}

The Haar mean then follows as in Lemma~\ref{lem:sm-haar-determinant-moments}.
We insert Eq.~\eqref{eq:sm-haar-state-moment} with $q=2k$, expand the even-parity projectors, and use Eq.~\eqref{eq:sm-exterior-algebra-character}.
With $Q_{\boldsymbol\epsilon}$ and $P_\pi$ defined as in Lemma~\ref{lem:sm-haar-determinant-moments} for $2k$ copies, this gives, for every integer $k\geq2$,
\begin{equation}
  \mu_k
  =\mathbb E_{\mathrm H}\nu_k
  =\frac{1}{2^{2k}(d)_{2k}}
  \sum_{\pi\in S_{2k}}
  \sum_{\boldsymbol\epsilon\in\{0,1\}^{2k}}
  \det\left(I_{2k}+Q_{\boldsymbol\epsilon}J_kP_\pi\right)^n.
  \label{eq:sm-general-k-exact-mean}
\end{equation}
Each pair $(\pi,\boldsymbol\epsilon)$ corresponds to the signed permutation $W=P_\pi Q_{\boldsymbol\epsilon}$, and $\det(I+Q_{\boldsymbol\epsilon}J_kP_\pi)=\det(I+J_kW)$.
We normalize these determinants by their value at the identity,
\begin{equation}
  \Delta_k(W)=\frac{\det(I_{2k}+J_kW)}{\det(I_{2k}+J_k)},
  \qquad
  \det(I_{2k}+J_k)=2^{2k-1},
  \label{eq:sm-delta-k-definition}
\end{equation}
where the last value follows from the eigenvalues $\pm i$ of $J_k$ in the plane of $\bm h_A$ and $\bm h_B$ and the eigenvalue $1$ elsewhere.
Since $2^{2k}(d)_{2k}=\prod_{j=0}^{2k-1}(D+2j)$ and $2^{(2k-1)n}/D^{2k}=D^{-1}$, Eq.~\eqref{eq:sm-general-k-exact-mean} becomes
\begin{equation}
  \mu_k
  =\frac{1}{D}\prod_{j=0}^{2k-1}\left(1+\frac{2j}{D}\right)^{-1}
  \sum_{W}\Delta_k(W)^n,
  \label{eq:sm-general-k-normalized-mean}
\end{equation}
where the sum runs over all $2^{2k}(2k)!$ signed permutations of the $2k$ copy labels.

\subsection{The leading determinant values}

At fixed $k$, the sum in Eq.~\eqref{eq:sm-general-k-normalized-mean} has finitely many terms, so its large-$N$ behavior is determined by the largest values of $\Delta_k(W)$.
We call the cycles of $\pi$, together with the signs of $Q_{\boldsymbol\epsilon}$ on their labels, the signed cycles of $W$.

\begin{lemma}[Determinant estimates]
  \label{lem:sm-determinant-estimates}
  Let $k\geq3$.
  For every signed permutation $W$ of the $2k$ copy labels, $\Delta_k(W)\geq0$, with $\Delta_k(W)=1$ only for $W=I_{2k}$.
  Exactly $k^2$ signed permutations have $\Delta_k(W)=r_k$.
  Each of them is a transposition of a label in batch $A$ with a label in batch $B$, with sign $+1$ on all other labels and, for each such pair, a unique choice of signs on the two transposed labels.
  Every other signed permutation satisfies
  \begin{equation}
    \Delta_k(W)\leq q_k.
    \label{eq:sm-general-k-channel-bound}
  \end{equation}
\end{lemma}

\begin{proof}
  The determinant is nonnegative because the nonreal eigenvalues of the orthogonal matrix $J_kW$ occur in conjugate pairs, every $+1$ eigenvalue contributes two, and any $-1$ eigenvalue makes the determinant zero.
  To evaluate it, we use the fact that $J_k$ differs from the identity only in the plane spanned by $\bm h_A$ and $\bm h_B$.
  We write $J_k=I+\mathsf H\mathsf K\mathsf H^{\top}$, where $\mathsf H$ is the $2k\times2$ matrix with columns $\bm h_A$ and $\bm h_B$, and $\mathsf K=\left(\begin{smallmatrix}-1&1\\-1&-1\end{smallmatrix}\right)$ represents $J_k-I$ in this plane.
  Then $I+J_kW=(I+W)+\mathsf H\mathsf K\mathsf H^{\top}W$ is a rank-two perturbation of $I+W$.

  We first assume that $I+W$ is invertible.
  A signed cycle of length $\ell$ with sign product $s$ contributes the factor $1-(-1)^\ell s\in\{0,2\}$ to $\det(I+W)$, so invertibility means that every signed cycle contributes $2$.
  Hence $\det(I+W)=2^{c(W)}$, where $c(W)$ is the number of signed cycles, including fixed labels.
  We measure the distance of $W$ from the identity by
  \begin{equation}
    \tau(W)=2k-c(W)=\sum_C(|C|-1),
    \label{eq:sm-signed-cycle-deficit}
  \end{equation}
  the minimal number of transpositions needed to write $\pi$, where the sum runs over the cycles $C$ of $W$.
  Since $W$ is orthogonal, the symmetric part of $W(I+W)^{-1}$ is $I/2$, and we write
  \begin{equation}
    W(I+W)^{-1}=\frac12I+\Omega,
    \qquad
    \Omega^{\top}=-\Omega.
    \label{eq:sm-signed-cayley-transform}
  \end{equation}
  With $t=\bm h_A^{\top}\Omega\bm h_B$, the matrix determinant lemma reduces the determinant to a two-dimensional one,
  \begin{equation}
    \det(I+J_kW)=2^{c(W)}\det\!\left[I_2+\mathsf K\begin{pmatrix}1/2&t\\-t&1/2\end{pmatrix}\right]=2^{c(W)+1}\left(\frac12-t\right)^2,
  \end{equation}
  and division by $2^{2k-1}$ gives
  \begin{equation}
    \Delta_k(W)=2^{2-\tau(W)}\left(\frac12-t\right)^2.
    \label{eq:sm-cycle-scalar-determinant}
  \end{equation}

  The size of $t$ is controlled by the cycle structure.
  On a signed cycle of length $\ell$ without a $-1$ eigenvalue,
  \begin{equation}
    (I+W)^{-1}=\frac12\sum_{j=0}^{\ell-1}(-W)^j.
    \label{eq:sm-signed-cycle-inverse}
  \end{equation}
  Since $\Omega=I/2-(I+W)^{-1}$, every off-diagonal entry of $\Omega$ within a cycle has absolute value $1/2$, and $\Omega$ has no entries between different cycles.
  As $t=k^{-1}\sum_{a\in A,\,b\in B}\Omega_{ab}$, a cycle $C$ contributes at most $|C\cap A|\,|C\cap B|/(2k)$, so that
  \begin{equation}
    |t|\leq\frac1{2k}\sum_C|C\cap A|\,|C\cap B|.
    \label{eq:sm-cycle-cayley-cut-bound}
  \end{equation}
  For a general bound, we use $|C\cap A|\,|C\cap B|\leq|C|^2/4$ and
  \begin{equation}
    \sum_{C\ \mathrm{nontrivial}}|C|^2\leq(\tau(W)+1)^2.
    \label{eq:sm-cycle-length-square-bound}
  \end{equation}
  The last inequality holds because replacing two lengths $\ell_1,\ell_2\geq2$ in the sum by the single length $\ell_1+\ell_2-1$ preserves $\sum_C(|C|-1)$ and increases the sum of squares by $2(\ell_1-1)(\ell_2-1)-1>0$.
  For $\tau(W)>0$, repeating this replacement leaves the single length $\tau(W)+1$.
  For $\tau(W)=0$, the sum is empty and the inequality holds trivially.
  Hence $|t|\leq(\tau(W)+1)^2/(8k)$.

  We now treat the invertible cases in order of $\tau(W)$.
  For $\tau(W)=0$, the identity is the only invertible signed permutation, and $\Delta_k(I_{2k})=1$.
  For $\tau(W)=1$, $W$ is a single signed transposition with sign $+1$ on all fixed labels.
  A transposition within a batch has $t=0$ and $\Delta_k(W)=1/2$.
  A transposition between the batches has two invertible sign choices, with $t=\pm1/(2k)$.
  The choice with $t=-1/(2k)$ gives $\Delta_k(W)=2[(k+1)/(2k)]^2=r_k$, and the other gives $(k-1)^2/(2k^2)<1/2$.
  These are the $k^2$ signed permutations with $\Delta_k(W)=r_k$, one for each pair of labels from different batches.
  For $\tau(W)=2,3,4$, we list the possible nontrivial cycle lengths and the largest value of the sum in Eq.~\eqref{eq:sm-cycle-cayley-cut-bound}.
  Equation~\eqref{eq:sm-cycle-scalar-determinant} then gives the bounds
  \begin{equation}
    \begin{array}{c|c|c|c}
      \tau(W) & \text{nontrivial cycle lengths} & \max\sum_C|C\cap A|\,|C\cap B| & \text{bound on }\Delta_k(W)             \\
      \hline
      2       & 3;\ 2+2                         & 2                              & \bigl(\tfrac12+\tfrac1k\bigr)^2         \\
      3       & 4;\ 3+2;\ 2+2+2                 & 4                              & \tfrac12\bigl(\tfrac12+\tfrac2k\bigr)^2 \\
      4       & 5;\ 4+2;\ 3+3;\ 3+2+2;\ 2+2+2+2 & 6                              & \tfrac14\bigl(\tfrac12+\tfrac3k\bigr)^2
    \end{array}.
    \label{eq:sm-small-deficit-bounds}
  \end{equation}
  For a given $k$, the admissible patterns have total length at most $2k$, and the bounds hold for each of them.
  The bound for $\tau(W)=2$ equals $(k+2)^2/(4k^2)\leq q_k$.
  For $\tau(W)=3$, the bound is $49/72<25/36=q_3$ at $k=3$ and at most $1/2$ for $k\geq4$.
  For $\tau(W)=4$, it is $9/16<25/36$ at $k=3$ and below $1/2$ for $k\geq4$.
  For $\tau(W)\geq5$, we use $|t|\leq(\tau(W)+1)^2/(8k)$ together with $k\geq(\tau(W)+1)/2$, which follows from $\tau(W)\leq2k-1$.
  This gives
  \begin{equation}
    \Delta_k(W)
    \leq2^{2-\tau(W)}\left(\frac{\tau(W)+3}{4}\right)^2
    \leq\frac12,
    \label{eq:sm-large-deficit-channel-bound}
  \end{equation}
  since the middle expression equals $1/2$ at $\tau(W)=5$ and decreases thereafter.

  It remains to treat signed permutations for which $I+W$ is singular, that is, with signed cycles on which $W$ has the eigenvalue $-1$.
  If their number is odd, then $\det(J_kW)=-1$, so the orthogonal matrix $J_kW$ itself has the eigenvalue $-1$ and $\Delta_k(W)=0$.
  If there are at least three such cycles, the rank-two perturbation cannot remove the kernel of $I+W$, and again $\Delta_k(W)=0$.
  In the remaining case, there are exactly two such cycles, of lengths $\ell_1$ and $\ell_2$.
  Let $\bm x_1$ and $\bm x_2$ be the corresponding normalized eigenvectors with eigenvalue $-1$, each supported on its cycle.
  Their nonzero coordinates are $\pm1/\sqrt{\ell_j}$, and we denote by $\alpha_j$ and $\beta_j$ the sums of these signs over the labels in batches $A$ and $B$, respectively.
  The kernel of $I+W$ is spanned by $\bm x_1$ and $\bm x_2$.
  In an orthonormal basis consisting of $\bm x_1$, $\bm x_2$, and a basis of their orthogonal complement, the normal matrix $I+W$ is block diagonal, with a zero block on the kernel and an invertible block on the complement.
  We expand $\det(I+J_kW)$ by multilinearity in the columns of this basis.
  The two columns along $\bm x_1$ and $\bm x_2$ must be taken from the perturbation $\mathsf H\mathsf K\mathsf H^{\top}W$, and all other columns from $I+W$, because any three columns of a rank-two matrix are linearly dependent.
  The determinant therefore factorizes into the determinant of the invertible block and that of the $2\times2$ matrix $[\bm x_1,\bm x_2]^{\top}\mathsf H\mathsf K\mathsf H^{\top}W[\bm x_1,\bm x_2]$.
  The invertible block contributes a factor $2$ from each cycle without the eigenvalue $-1$ and the product $\ell_j$ of the nonzero eigenvalues of $I+W$ on the cycle of length $\ell_j$.
  Since $W\bm x_j=-\bm x_j$ and $\det\mathsf K=2$, the $2\times2$ determinant is $2\det(\mathsf H^{\top}[\bm x_1,\bm x_2])^2=2(\alpha_1\beta_2-\alpha_2\beta_1)^2/(k^2\ell_1\ell_2)$.
  Therefore
  \begin{equation}
    \Delta_k(W)
    =2^{-\tau(W)}\frac{(\alpha_1\beta_2-\alpha_2\beta_1)^2}{k^2}.
    \label{eq:sm-two-null-cycle-determinant}
  \end{equation}
  Since the two batches are disjoint, $|\alpha_j|+|\beta_j|\leq\ell_j$, and hence $|\alpha_1\beta_2-\alpha_2\beta_1|\leq\ell_1\ell_2$.
  Together with $\tau(W)\geq(\ell_1-1)+(\ell_2-1)$, this gives
  \begin{equation}
    \Delta_k(W)
    \leq\frac1{k^2}
    \prod_{j=1}^{2}\frac{\ell_j^2}{2^{\ell_j-1}}
    \leq\frac{81}{16k^2},
    \label{eq:sm-null-cycle-channel-bound}
  \end{equation}
  where $\ell^2/2^{\ell-1}$ is maximized over positive integers by $\ell=3$, with value $9/4$.
  The final bound is below $25/36$ for $k=3$ and below $1/2$ for $k\geq4$.
  In all cases other than the identity and the $k^2$ transpositions above, we have thus shown $\Delta_k(W)\leq q_k$.
\end{proof}

\subsection{Proof of Theorem~\ref{thm:sm-general-k-haar}}

By Lemma~\ref{lem:sm-determinant-estimates}, the sum in Eq.~\eqref{eq:sm-general-k-normalized-mean} equals $1+k^2r_k^n$ plus at most $2^{2k}(2k)!$ further terms, each bounded by $q_k^n$.
The prefactor satisfies $\prod_j(1+2j/D)^{-1}=1+\mathcal{O}(2^{-n})$, and this relative correction is also included in $\mathcal{O}(q_k^n)$ because $q_k\geq1/2$.
This proves Eq.~\eqref{eq:sm-general-k-asymptotic}.

The mean purity also controls the MRE, without a separate second-moment calculation, because the purity has a lower bound that holds for every input.
Since $\mathcal C_k(\rho)$ acts on a $D$-dimensional space, $D\nu_k\geq1$, which gives $M_k\leq\ln D/(k-1)$ for every state.
Together with $\ln x\leq x-1$, this gives
\begin{equation}
  0\leq\frac{\ln D}{k-1}-\mathbb E_{\mathrm H}M_k
  =\frac{\mathbb E_{\mathrm H}\ln(D\nu_k)}{k-1}
  \leq\frac{D\mu_k-1}{k-1}=\mathcal{O}(r_k^n).
\end{equation}
The maximum of $M_k$ lies between this mean and $\ln D/(k-1)$, and division by $N$ proves Eq.~\eqref{eq:sm-general-k-density}.
Finally, the inequality between the trace norm and the Hilbert-Schmidt norm gives
\begin{equation}
  \mathbb E_{\mathrm H}\|\mathcal C_k(\rho)-I/D\|_1^2
  \leq D\mu_k-1=\mathcal{O}(r_k^n),
\end{equation}
and Markov's inequality implies convergence in probability.

\section{Derivations and additional checks for the SYK model}
\label{sec:sm-syk-methods}

In this section, we derive the effective action and the SD equations of the main text and classify their saddles by copy-permutation symmetry.
We then examine the two assumptions underlying the MRE density $m(\beta,t)$, namely that the two saddles dominate $\mathbb E_JZ_R$ and that the annealed and quenched averages agree.
We prove the second assumption for the normalization $Z_B$ and test both assumptions numerically by searching for other saddles and by computing the MRE directly at finite $N$.
The reduction of the SD equations to two scalar functions and the numerical method used to solve them are given in the End Matter.

\subsection{Effective action and SD equations}

The effective action in Eq.~\eqref{eq:syk-action} contains a quartic term from the disorder average, a constraint term enforcing the definition of $G$, and a determinant term from integrating out the fermions.
We derive these terms in turn, using the contour and copy basis of the main text with the same symbol $\gamma_i$ and normalization for the Majorana operators and their Grassmann fields.
All double integrals below are taken over $[0,2]^2$, and all transposes exchange the contour coordinates and copy indices without complex conjugation.

The disorder average couples the four copies through the term quartic in $\widehat G$.
For each index set $\mathcal A=\{i_1<i_2<i_3<i_4\}$, we define $X_{\mathcal A}=\sum_a\int\dd u\,\dot z(u)\gamma_{i_1}^{(a)}(u)\gamma_{i_2}^{(a)}(u)\gamma_{i_3}^{(a)}(u)\gamma_{i_4}^{(a)}(u)$, so that the interaction contributes $\exp(\sum_{\mathcal A}J_{\mathcal A}X_{\mathcal A})$.
Since the couplings are independent Gaussian variables of variance $6J^2/N^3$, the disorder average gives
\begin{equation}
  \mathbb E_J e^{\sum_{\mathcal A}J_{\mathcal A}X_{\mathcal A}}
  =\exp\!\left(\frac{3J^2}{N^3}\sum_{\mathcal A}X_{\mathcal A}^2\right)
  =\exp\!\left[\frac{NJ^2}{8}\int\dd u\,\dd v\,\dot z(u)\dot z(v)\sum_{a,b}\widehat G_{ab}(u,v)^4\right].
  \label{eq:sm-syk-disorder-average}
\end{equation}
For the second equality, we move each field at $v$ in $X_{\mathcal A}^2$ next to the field with the same Majorana index at $u$, which requires six fermionic interchanges and gives a positive sign.
The resulting bilinears commute, and summing over $\mathcal A$ reproduces $(N\widehat G_{ab})^4/4!$ because terms with repeated indices vanish and each set of four distinct indices occurs $4!$ times.
The sum includes every pair of copies $a,b$ because all copies share the same couplings.

Before integrating out the fermions, we absorb the reflection $\mathcal R$ into the interaction by rotating the fields.
We write $\widetilde\gamma_i(u)=O(u)^{\top}\gamma_i(u)$ and $G_{ab}(u,v)=N^{-1}\sum_i\widetilde\gamma_i^{(a)}(u)\widetilde\gamma_i^{(b)}(v)$, so that $\widehat G(u,v)=O(u)G(u,v)O(v)^{\top}$.
The insertion of $\mathcal R$ imposes the matching condition $\gamma_i(1^+)=R\gamma_i(1^-)$, which becomes $\widetilde\gamma_i(1^+)=\widetilde\gamma_i(1^-)$ after this rotation.
The orthogonal matrix $O(u)$ also mixes annihilation operators among the copies without mixing them with creation operators, so it leaves the product Fock vacuum invariant.
The boundary conditions at the contour endpoints are therefore unchanged.
Since $O(u)$ is constant on each contour segment and the rotated fields are continuous at $u=1$, their kinetic term is $S_{\mathrm{kin}}=\tfrac14\sum_{i,a}\int\dd u\,\widetilde\gamma_i^{(a)}\partial_u\widetilde\gamma_i^{(a)}$ for the normalization $\{\gamma_i,\gamma_j\}=2\delta_{ij}$.
Consequently, the explicit dependence on $R$ in the effective action appears only through $\widehat G$ in the quartic term.

The free propagator of the rotated fields is fixed by the boundary state and carries the KM boundary condition.
For each pair of Majorana operators, the complex fermion $f_m=(\gamma_{2m-1}-i\gamma_{2m})/2$ obeys $\bra B f_mf_m^\dagger\ket B=1$ and $\bra B f_m^\dagger f_m\ket B=0$.
The free particle and hole kernels are therefore $F_0=\Theta I_4$ and $-F_0^{\top}$.
Since $F_0=[\tfrac12\operatorname{sgn}(u-v)+\tfrac12]I_4$, the boundary condition resides in the symmetric part $\tfrac12I_4$, and replacing $F_0$ by its antisymmetric part before solving Dyson's equation would discard it.

We now impose the definition of $G$ and integrate out the fermions.
We insert the identity
\begin{equation}
  1=\int\mathcal DG\,\mathcal D\Sigma\,\exp\!\left[-\frac N2\int\dd u\,\dd v\sum_{a,b}\Sigma_{ab}(u,v)\left(G_{ab}(u,v)-\frac1N\sum_i\widetilde\gamma_i^{(a)}(u)\widetilde\gamma_i^{(b)}(v)\right)\right],
  \label{eq:sm-syk-constraint}
\end{equation}
where the integration contour of the antisymmetric multiplier $\Sigma$ implements the functional delta constraint and the normalization of this representation is included in the measures.
The second term in the exponent is quadratic in the fermions and shifts the inverse free propagator for each Majorana operator by $-\Sigma$, corresponding to a shift of $-2\Sigma$ for the complex fermion.
Relative to the free integral, the Pfaffian for one Majorana pair is
\begin{equation}
  \left[\det(1-2F_0\Sigma)\det(1+2F_0^{\top}\Sigma)\right]^{1/2}
  =\det(1-2F_0\Sigma),
  \label{eq:sm-syk-pair-determinant}
\end{equation}
since $\Sigma^{\top}=-\Sigma$ implies $\det(1+2F_0^{\top}\Sigma)=\det(1-2\Sigma F_0)=\det(1-2F_0\Sigma)$, with the phase continued from $\Sigma=0$.
The $N/2$ pairs thus contribute $\det(1-2F_0\Sigma)^{N/2}$, which combines with the constraint and the disorder average to give $\mathbb E_JZ_R=\int\mathcal DG\,\mathcal D\Sigma\,e^{-NI[G,\Sigma]}$ with the action in Eq.~\eqref{eq:syk-action}.
Field-independent factors in the measure are fixed by the normalization $Z_R=1$ at $J=0$, where $\mathcal R$ leaves the product vacuum invariant.

Varying this action with respect to $\Sigma$ relates $G$ to $\Sigma$ through Dyson's equation, and varying it with respect to $G$ gives the self-energy equation.
We write $\Tr$ for the trace over copy indices and contour coordinates, $\Tr(AB)=\int\dd u\,\dd v\sum_{a,b}A_{ab}(u,v)B_{ba}(v,u)$, and define the particle kernel dressed by $\Sigma$ as $F=(1-2F_0\Sigma)^{-1}F_0$, which is Eq.~\eqref{eq:syk-dyson}.
The determinant term varies as $\delta\ln\det(1-2F_0\Sigma)=\Tr[(1-2F_0\Sigma)^{-1}(-2F_0\,\delta\Sigma)]=-2\Tr(F\,\delta\Sigma)$.
Using $\delta\Sigma_{ba}(v,u)=-\delta\Sigma_{ab}(u,v)$ to write $\Tr(F\,\delta\Sigma)$ as a double integral, we obtain
\begin{equation}
  \delta_\Sigma I
  =\Tr(F\,\delta\Sigma)+\frac12\int\dd u\,\dd v\sum_{a,b}\delta\Sigma_{ab}(u,v)G_{ab}(u,v)
  =\frac12\int\dd u\,\dd v\sum_{a,b}\delta\Sigma_{ab}(u,v)\bigl[G_{ab}(u,v)-2F_{ab}(u,v)\bigr].
  \label{eq:sm-syk-sigma-variation}
\end{equation}
Since $\delta\Sigma$ is antisymmetric, only the antisymmetric part of the bracket contributes.
With $G^{\top}=-G$, this part is $G-(F-F^{\top})$, and stationarity gives $G=F-F^{\top}$, which is Eq.~\eqref{eq:syk-majorana-propagator}.
Here, $-F^{\top}=(1+2F_0^{\top}\Sigma)^{-1}(-F_0^{\top})$ is the dressed hole kernel, so $G$ is the sum of the particle and hole contributions.

For the variation with respect to $G$, only the constraint term and the quartic term contribute.
The relation $\widehat G(u,v)=O(u)G(u,v)O(v)^{\top}$ gives $\delta\widehat G_{ab}(u,v)=\sum_{c,d}O_{ac}(u)\,\delta G_{cd}(u,v)\,O_{bd}(v)$.
Since varying $\widehat G_{ab}^4$ produces $4\widehat G_{ab}^3\,\delta\widehat G_{ab}$, moving the two factors of $O$ onto $\widehat G^{\circ3}$ yields
\begin{equation}
  \delta_GI
  =\frac12\int\dd u\,\dd v\sum_{a,b}\delta G_{ab}(u,v)\Bigl\{\Sigma_{ab}(u,v)-J^2\dot z(u)\dot z(v)\bigl[O(u)^{\top}\widehat G(u,v)^{\circ3}O(v)\bigr]_{ab}\Bigr\}.
  \label{eq:sm-syk-g-variation}
\end{equation}
Both terms in the braces are antisymmetric, because $\Sigma$ and $\widehat G$ are antisymmetric and the elementwise cube preserves this property.
Stationarity with respect to $G$ therefore gives Eq.~\eqref{eq:syk-self-energy}, which completes the derivation of the SD equations.

\subsection{Permutation symmetry and the eight equivalent saddles}

The SD equations are invariant under permutations of the four copies, and a saddle with $h\neq0$ breaks this symmetry down to $C_3$.
To see the invariance, let $P$ be a permutation matrix acting on the copy indices.
Since $P$ commutes with $R=I_4-\bm1\bm1^{\top}/2$, it also commutes with the change of basis $O(u)$, so $G\mapsto PGP^{\top}$ induces $\widehat G\mapsto P\widehat GP^{\top}$.
The free propagator $F_0=\Theta I_4$ is proportional to $I_4$, so this transformation, together with $\Sigma\mapsto P\Sigma P^{\top}$, leaves each term of the action in Eq.~\eqref{eq:syk-action} unchanged and maps solutions of the SD equations to solutions.

The symmetric saddle $G=gI_4$ is invariant under every permutation, whereas a saddle $G=gI_4+hA$ with $h\neq0$ is invariant only under permutations with $PAP^{\top}=A$.
The matrix $A$ singles out copy 1, whose row sum $3$ differs from the row sums $-1$ of the other copies.
It also orients the remaining copies, because its positive entries $A_{23}=A_{34}=A_{42}=1$ form the directed cycle $2\to3\to4\to2$.
A permutation with $PAP^{\top}=A$ must therefore fix copy 1 and preserve this cycle, which leaves only the three cyclic permutations of copies 2, 3, and 4.
The residual symmetry group is thus $C_3$, and acting with all permutations on one saddle gives $|S_4|/|C_3|=24/3=8$ distinct saddles with equal action.
Equivalently, each saddle with $h\neq0$ is labeled by a distinguished copy and an orientation of the other three copies, which gives $4\times2=8$ saddles.

The same symmetry pattern appears in the leading terms of the Haar mean of the purity.
For large $N$, the mean purity in Eq.~\eqref{eq:sm-haar-exact-mean-sm} is dominated by $8\cdot3^N=8\cdot9^n$, which comes from the eight pairs $(\pi,\boldsymbol\epsilon)$ with determinant 9 in Eq.~\eqref{eq:sm-first-determinant-histogram}.
In each of these pairs, $\pi$ is a three-cycle with one fixed label, and $Q_{\boldsymbol\epsilon}$ is negative on the three-cycle and positive on the fixed label.
Each term is therefore labeled by a distinguished copy, namely the fixed label, and an orientation of the three-cycle on the other copies, exactly as for the saddles with $h\neq0$.
Relabeling the copies permutes these eight terms, and the permutations that leave a given term unchanged again form $C_3$.
The saddles with broken symmetry in the SYK model thus have the same symmetry pattern as the leading terms of the Haar mean of the purity.

\subsection{Annealed and quenched averages of the KM normalization}
\label{sec:sm-syk-normalization}

We now show that $\ln\mathbb E_J Z_B$ and $\mathbb E_J\ln Z_B$ agree to leading order in $N$ at fixed $\beta J$.
\begin{proposition}[Self-averaging of the KM normalization]
  \label{prop:sm-zb-averaging}
  For the Gaussian SYK ensemble of Eq.~\eqref{eq:syk-hamiltonian},
  \begin{equation}
    0\leq\ln\mathbb E_J Z_B-\mathbb E_J\ln Z_B
    \leq C_n+c\frac{(\beta J)^2}{N},
    \qquad n=N/2,
    \label{eq:sm-zb-averaging-bound}
  \end{equation}
  where $c$ is independent of $N$ and $\beta$ for sufficiently large $N$, and
  \begin{equation}
    C_n=\sum_{k=1}^{n-1}\bigl[\psi(2k)-\psi(k)-\ln2\bigr]
    =\frac14\ln n+\mathcal{O}(1).
    \label{eq:sm-gaussian-log-constant}
  \end{equation}
  Here, $\Gamma(x)$ is the gamma function, $\psi(x)=\mathrm d\ln\Gamma(x)/\mathrm dx$ is the digamma function, and $C_1=0$.
  In particular, the right-hand side of Eq.~\eqref{eq:sm-zb-averaging-bound} is $o(N)$ at fixed $\beta J$.
  Moreover, $\mathbb E_JZ_B=2^{-N/2}\mathbb E_J\operatorname{Tr}e^{-\beta H}$.
\end{proposition}

The proof rests on the exact factorization $Z_B=Z_+p_B$, where $Z_+=\Tr_{\mathcal H_+}e^{-\beta H}$ is the partition function in the even-parity sector $\mathcal H_+$ of dimension $d=2^{n-1}$, and $p_B=\bra B\rho_+\ket B$ is the probability of the boundary state in the thermal state $\rho_+=e^{-\beta H}|_{\mathcal H_+}/Z_+$.
We control the fluctuations of $Z_+$ by the concentration bound of Anschuetz et al.~\cite{anschuetz2025strongly} and those of $p_B$ by a bound for projection onto a random Gaussian state.

For the first factor, we write $H=-\sigma\sum_\alpha\xi_\alpha\gamma_\alpha$, where $\alpha$ runs over the $\binom N4$ index sets $i<j<k<l$, $\gamma_\alpha=\gamma_i\gamma_j\gamma_k\gamma_l$, $\sigma^2=6J^2/N^3$, and the $\xi_\alpha$ are independent standard Gaussian variables.
The commutation index of these monomials satisfies~\cite{anschuetz2025strongly}
\begin{equation}
  \Delta_4\equiv\sup_\rho\frac{1}{\binom N4}\sum_\alpha\bigl(\Tr\rho\gamma_\alpha\bigr)^2=\mathcal{O}(N^{-2}),
  \label{eq:sm-syk-commutation-index}
\end{equation}
where the supremum is over density operators on the full Fock space.
Every $\gamma_\alpha$ preserves parity, so cyclicity of the trace on $\mathcal H_+$ gives the first relation below.
Since $\rho_+$, extended by zero to the odd-parity sector, is a density operator on the full Fock space, Eq.~\eqref{eq:sm-syk-commutation-index} then bounds the gradient,
\begin{equation}
  \partial_{\xi_\alpha}\ln Z_+=\beta\sigma\Tr(\rho_+\gamma_\alpha),
  \qquad
  \|\nabla_\xi\ln Z_+\|_2^2\leq\beta^2\sigma^2\binom N4\Delta_4
  =\mathcal{O}\!\left(\frac{(\beta J)^2}{N}\right).
\end{equation}
The Gaussian concentration and exponential-moment estimates used in Ref.~\cite{anschuetz2025strongly} therefore imply
\begin{equation}
  0\leq\ln\mathbb E_J Z_+-\mathbb E_J\ln Z_+
  \leq c\frac{(\beta J)^2}{N}.
  \label{eq:sm-parity-thermal-bound}
\end{equation}

For the second factor, we use rotation invariance to replace the fixed boundary state by a random Gaussian state.
Let $U$ be a Gaussian unitary implementing a Haar-distributed rotation in $SO(N)$ of the Majorana operators, and let $\ket{B_U}=U\ket B$.
The sign ambiguity in the lift to $\operatorname{Spin}(N)$ does not affect the projector onto $\ket{B_U}$.
The following bound holds for any state of even parity, and we apply it to $\rho_+$ below.

\begin{lemma}[Overlap with a random Gaussian state]
  \label{lem:sm-gaussian-overlap}
  For any density operator $\rho$ supported on $\mathcal H_+$, the probability $p_U=\bra{B_U}\rho\ket{B_U}$ satisfies
  \begin{equation}
    0\leq\mathbb E_U[-\ln(dp_U)]\leq C_n.
    \label{eq:sm-gaussian-log-overlap}
  \end{equation}
\end{lemma}

\begin{proof}
  Irreducibility of the even spinor representation implies $\mathbb E_U\ketbra{B_U}=I_{\mathcal H_+}/d$, so $\mathbb E_Up_U=1/d$, and the lower bound follows from Jensen's inequality.

  For the upper bound, we construct a Haar-distributed oriented orthonormal frame of Majorana directions one pair at a time.
  Projection onto $\ket{B_U}$ is equivalent to selecting the empty outcome for each complex fermion mode associated with this frame.
  Only the frame is random here, and the outcomes are always chosen to be empty rather than sampled with their Born probabilities.
  After the previous empty outcomes have been selected, suppose that $r\geq2$ modes remain, with normalized conditional state $\rho_r$.
  This state has even parity but need not be Gaussian, and the remaining frame is conditionally Haar-distributed in the orthogonal complement of the directions already chosen.
  For the next orthonormal pair $u,v\in\mathbb R^{2r}$, we define $\gamma(u)=\sum_a u_a\gamma_a$ and the covariance matrix $(\Gamma_r)_{ab}=(i/2)\Tr\rho_r[\gamma_a,\gamma_b]$.
  In the convention $f_m=(\gamma_{2m-1}-i\gamma_{2m})/2$ of the main text, the conditional probability of an empty outcome is
  \begin{equation}
    q_r=\Tr\rho_r\frac{1+i\gamma(u)\gamma(v)}2
    =\frac{1+u^{\top}\Gamma_r v}{2}.
    \label{eq:sm-conditional-empty-probability}
  \end{equation}

  We now bound $\mathbb E[-\ln q_r]$ uniformly in the conditional state.
  The covariance matrix of any fermionic state satisfies $\|\Gamma_r\|_{\mathrm{op}}\leq1$, and since $\Gamma_r$ is antisymmetric, $\Gamma_ru$ is orthogonal to $u$.
  For fixed $u$, the vector $v$ is uniform on the unit sphere in $u^\perp$, so $u^{\top}\Gamma_r v=-(\Gamma_ru)^{\top}v$ has the same distribution as $aX$, where $a=\|\Gamma_r u\|\leq1$ and $X$ is one coordinate of a uniform point on $S^{2r-2}$.
  Symmetry of the distribution of $X$ gives
  \begin{equation}
    \mathbb E[-\ln(1+aX)]
    =-\frac12\mathbb E\ln(1-a^2X^2)
    \leq-\frac12\mathbb E\ln(1-X^2).
  \end{equation}
  Thus $\mathbb E[-\ln q_r]$ is largest at $a=1$, in which case $q_r=(1+X)/2$ has the beta distribution with parameters $(r-1,r-1)$.
  Averaging over the next pair, conditionally on all previously chosen directions, yields
  \begin{equation}
    \mathbb E[-\ln q_r\mid\text{previous directions}]
    \leq\psi(2r-2)-\psi(r-1).
  \end{equation}
  This bound also ensures logarithmic integrability when $\rho$ is not full rank.

  The probability of the full projection is the product of these conditional probabilities.
  Since the frame is oriented and every selected outcome is empty, even total parity forces the last mode to be empty with probability one.
  Therefore $p_U=\prod_{r=2}^{n}q_r$, and successive conditional averaging gives
  \begin{equation}
    \mathbb E_U[-\ln p_U]
    \leq\sum_{r=2}^{n}\bigl[\psi(2r-2)-\psi(r-1)\bigr]
    =\ln d+C_n,
  \end{equation}
  that is, $\mathbb E_U[-\ln(dp_U)]\leq C_n$.
  Finally, the duplication formula for the digamma function gives
  \begin{equation}
    \psi(2k)-\psi(k)-\ln2
    =\frac12\bigl[\psi(k+\tfrac12)-\psi(k)\bigr]
    =\frac1{4k}+\mathcal{O}(k^{-2}),
  \end{equation}
  which proves the asymptotic expression in Eq.~\eqref{eq:sm-gaussian-log-constant}.
\end{proof}

\begin{proof}[Proof of Proposition~\ref{prop:sm-zb-averaging}]
  The distribution of the couplings is invariant under $SO(N)$ rotations of the Majorana operators, which are implemented by parity-preserving Gaussian unitaries and leave $Z_+$ unchanged.
  Averaging the boundary state over such rotations therefore gives
  \begin{equation}
    \mathbb E_J Z_B=\frac{\mathbb E_J Z_+}{d},
    \qquad
    \mathbb E_J\ln p_B
    =\mathbb E_J\mathbb E_U\ln\bra{B_U}\rho_+(H)\ket{B_U}.
  \end{equation}
  Together with $Z_B=Z_+p_B$, the difference of logarithmic averages decomposes as
  \begin{equation}
    \ln\mathbb E_J Z_B-\mathbb E_J\ln Z_B
    =\ln\mathbb E_J Z_+-\mathbb E_J\ln Z_+
    +\mathbb E_J[-\ln(dp_B)].
  \end{equation}
  Equation~\eqref{eq:sm-parity-thermal-bound} bounds the first term, and Lemma~\ref{lem:sm-gaussian-overlap} applied to $\rho=\rho_+(H)$ bounds the second, which proves Eq.~\eqref{eq:sm-zb-averaging-bound}.
  Since the ensemble is also invariant under improper rotations, which exchange the two parity sectors, we further have $\mathbb E_JZ_+=\tfrac12\mathbb E_J\operatorname{Tr}e^{-\beta H}$, and hence $\mathbb E_JZ_B=2^{-N/2}\mathbb E_J\operatorname{Tr}e^{-\beta H}$.
\end{proof}

\subsection{Numerical checks of saddle dominance}

We ask whether a saddle outside the ansatz $G=gI_4+hA$ could have an action with a smaller real part than the two candidate saddles, with $h=0$ and $h\neq0$.
We solve the SD equations for the full $4\times4$ matrix kernels without imposing the ansatz, starting from the free propagator, mixtures of the two candidates, and real or complex perturbations of either candidate.
The perturbations include components outside the algebra generated by $I_4$ and $A$, and we do not project onto this algebra during the iteration.
We perform this search at six parameter points on both sides of the phase boundary.
The point $(\beta J,Jt)=(0,0.5425)$ lies near the crossing at $\beta=0$, and $(3.22,1.2292)$ lies near the point where the phase boundary reaches its largest value of $Jt$.
The pairs $(5.85,0)$, $(5.91,0)$ and $(4,1.16)$, $(4,1.195)$ lie on either side of the crossings at $t=0$ and at $\beta J=4$, respectively.
At each point, we solve the equations from 96 initial guesses on each of the grids $L=96,128$.
We retain a converged solution only if the largest absolute entry of the residual of the full SD equations is below $10^{-9}$.
Most converged solutions reproduce the action of one of the two candidates, and the rest have actions with larger real parts.
Among the latter is a third solution within the ansatz $G=gI_4+hA$.
These results support the assumption that the two candidates dominate in the region explored, although a finite search cannot exclude other stationary points or determine whether they contribute to the integral.

We also check that the transition is a crossing of two distinct solutions, as expected for a first-order transition, rather than a point where one solution branches off from the other.
If another stationary solution branched off from a candidate, an infinitesimal deformation toward it would be reproduced when the self-energy equation and Dyson's equation are applied in turn.
In the original copy basis, we write $\widehat F(u,v)=O(u)F(u,v)O(v)^{\top}$ and $\widehat\Sigma(u,v)=O(u)\Sigma(u,v)O(v)^{\top}$.
We denote the symmetric part of $\widehat F$ by $S=(\widehat F+\widehat F^{\top})/2$.
Linearizing the two equations around a solution gives, for an antisymmetric kernel $X$,
\begin{equation}
  (\mathcal KX)_{ab}(u,v)=3J^2\dot z(u)\dot z(v)\widehat G_{ab}(u,v)^2X_{ab}(u,v),
  \qquad \mathcal DX=\widehat G X\widehat G+4SXS,
  \label{eq:sm-syk-linear-response}
\end{equation}
where products in $\mathcal D$ compose kernels over contour coordinates and copy indices.
The expression for $\mathcal K$ follows by differentiating the self-energy equation.
For $\mathcal D$, varying Dyson's equation gives $\delta\widehat F=2\widehat F\,\delta\widehat\Sigma\,\widehat F$, and antisymmetry of $\delta\widehat\Sigma$ implies
\begin{equation}
  \delta\widehat G=2\bigl(\widehat F\,\delta\widehat\Sigma\,\widehat F+\widehat F^{\top}\delta\widehat\Sigma\,\widehat F^{\top}\bigr)=\widehat G\,\delta\widehat\Sigma\,\widehat G+4S\delta\widehat\Sigma S,
\end{equation}
where we used $\widehat F=S+\widehat G/2$, so that the cross terms between $S$ and $\widehat G$ cancel.
An infinitesimal deformation $X=\delta\widehat G$ toward another stationary solution must therefore satisfy $X=\mathcal D\mathcal KX$, so a new branch can split off only where $\mathcal D\mathcal K$ has the eigenvalue one.

Numerically, $\mathcal D\mathcal K$ has no eigenvalue equal to one for either candidate at the crossings.
We compute the eigenvalues with the largest real parts at the extrapolated crossing points $Jt=0.54249$ at $\beta=0$, $\beta J=5.8784$ at $t=0$, and $Jt=1.1779$ at $\beta J=4$, allowing general antisymmetric perturbations in the four-copy space.
For all of $L=128,160,256$, the largest real part remains below one.
Together with the finite jump in $h$, this indicates that the two candidates are distinct and locally isolated at the crossing.

\subsection{Self-averaging and annealed averages at finite \texorpdfstring{$N$}{N}}

Using exact diagonalization at finite $N$, we test whether the MRE density is self-averaging and whether the annealed and quenched averages agree, without relying on the saddle-point approximation.
For each independent realization, we diagonalize the Hamiltonian in the even-parity sector, prepare the normalized KM state, and calculate the purity of its two-copy convolution.
At $t=0$, we sample $N=12,14,\ldots,26$ at $\beta J=4,6,8$, using 32 realizations for each $N\leq20$ and 16, 12, and 6 realizations for $N=22,24,26$, respectively.
At $\beta J=4$, we also sample $N=12,16,20,24,26$ at $Jt=0,0.9,1.1779,1.4,1.8$, where $Jt=1.1779$ is near the crossing, using the same numbers of realizations for each $N$ as at $t=0$.
We record both $\ln Z_B$ and $\ln Z_R=\ln\nu_2+4\ln Z_B$, so that averages of logarithms and logarithms of averages can be evaluated from the same samples.

To test self-averaging, we estimate the relative variance of $M_2/N$ by the unbiased sample variance divided by the squared sample mean.
Near the transition, the relative variance decreases overall as $N$ increases, which is consistent with self-averaging, although fewer realizations are available at the largest sizes.

The agreement of the annealed and quenched averages does not follow from self-averaging, because rare realizations can dominate an exponential average such as $\mathbb E_JZ_R$.
For the MRE density, the difference between the two averages is
\begin{equation}
  \frac{\mathbb E_JM_2}{N}-\frac{-\ln\mathbb E_JZ_R+4\ln\mathbb E_JZ_B}{N}
  =\frac{\ln\mathbb E_JZ_R-\mathbb E_J\ln Z_R-4(\ln\mathbb E_JZ_B-\mathbb E_J\ln Z_B)}{N}.
  \label{eq:sm-mre-averaging-difference}
\end{equation}
Proposition~\ref{prop:sm-zb-averaging} shows that the contribution of $Z_B$ to Eq.~\eqref{eq:sm-mre-averaging-difference} vanishes as $N\to\infty$ at fixed $\beta J$, so the difference is controlled by $(\ln\mathbb E_JZ_R-\mathbb E_J\ln Z_R)/N$.
We estimate the difference in Eq.~\eqref{eq:sm-mre-averaging-difference} from the sampled realizations, evaluating the exponential averages in logarithmic form to avoid numerical overflow.
Near the transition, the estimated difference shows an overall decrease as $N$ increases, which supports the condition $\ln\mathbb E_JZ_R-\mathbb E_J\ln Z_R=o(N)$ for the parameters examined.

\end{document}